\documentclass[12pt,onecolumn, draftclsnofoot]{IEEEtran}
\usepackage{bbm}
\usepackage{mathrsfs}
\usepackage{amsfonts}
\usepackage{epsfig}
\usepackage{psfrag}
\usepackage{graphicx}
\usepackage{epstopdf}
\usepackage{cite}
\usepackage{caption}

\usepackage{array}
\usepackage{stfloats}
\usepackage{textcomp}
\usepackage{verbatim}
\usepackage[caption=false,font=normalsize,labelfont=sf,textfont=sf]{subfig}
\usepackage[justification=centering]{caption}
\usepackage{multirow}
\usepackage{makecell}

\usepackage{amssymb}
\usepackage{amsthm}
\usepackage{algorithm,algorithmic}
\usepackage[tbtags]{amsmath}
\usepackage{url}
\usepackage{bm}

\usepackage{color}
\definecolor{gray}{RGB}{128,128,128}

\allowdisplaybreaks

\newtheorem{theorem}{Theorem}
\newtheorem{assumption}{Assumption}

\newtheorem{lemma}{Lemma}

\newtheorem{remark}{Remark}

\renewcommand{\IEEEproof}{\textbf{Proof}.}

\usepackage{color}

\definecolor{gray}{RGB}{128,128,128}

\allowdisplaybreaks

\begin{document}
\title{\textbf{Reduced Network Cumulative Constraint Violation for Distributed Bandit Convex Optimization under Slater’s Condition}}
\author{Kunpeng~Zhang,
        Xinlei~Yi,
        Jinliang~Ding,
        Ming~Cao,\\
        Karl~H.~Johansson,
        and Tao~Yang
\thanks{K. Zhang, J. Ding and T. Yang are with the State Key Laboratory of Synthetical Automation for Process Industries, Northeastern University, Shenyang 110819, China {\tt\small 2110343@stu.neu.edu.cn; \{jlding; yangtao\}@mail.neu.edu.cn}}%
\thanks{X. Yi is with Department of Control Science and Engineering, College of Electronics and Information Engineering, Tongji University, Shanghai 201800, China {\tt\small xinleiyi@tongji.edu.cn}}%
\thanks{M. Cao is with the Engineering and Technology Institute Groningen, Faculty of Science and Engineering, University of Groningen, AG 9747 Groningen, The Netherlands {\tt\small m.cao@rug.nl}}%
\thanks{K. H. Johansson is with Division of Decision and Control Systems, School of Electrical Engineering and Computer
Science, KTH Royal Institute of Technology, and he is also affiliated with Digital Futures, 10044, Stockholm, Sweden {\tt\small kallej@kth.se}}%
}


\maketitle

\begin{abstract}
This paper studies the distributed bandit convex optimization problem with time-varying inequality constraints, where the goal is to minimize network regret and cumulative constraint violation. 
To calculate network cumulative constraint violation,
existing distributed bandit online algorithms solving this problem directly use the clipped constraint function to replace its original constraint function.
However, the use of the clipping operation renders Slater’s condition (i.e, there exists a point that strictly satisfies the inequality constraints at all iterations) ineffective to achieve reduced network cumulative constraint violation.
To tackle this challenge, we propose a new distributed bandit online primal--dual algorithm. If local loss functions are convex, we show that the proposed algorithm establishes an $\mathcal{O}\big( {{T^{\max \{ {c,1-c} \}}}} \big)$ network regret bound and an $\mathcal{O}( {{T^{1 - c/2}}} )$ network cumulative constraint violation bound, where $T$ is the total number of iterations and $c \in ( {0,1} )$  is a user-defined trade-off parameter. When Slater’s condition holds, the network cumulative constraint violation bound is reduced to $\mathcal{O}( {{T^{1 - c}}} )$.
In addition, if local loss functions are strongly convex, for the case where strongly convex parameters are unknown, the network regret bound is reduced to $\mathcal{O}( {{T^{1 - c}}} )$, and the network cumulative constraint violation bound is $\mathcal{O}( {{T^{1 - c/2}}} )$ and $\mathcal{O}( {{T^{1 - c}}} )$ without and with Slater’s condition, respectively. For the case where strongly convex parameters are known, the network regret bound is further reduced to $\mathcal{O}\big( {\log ( T )} \big)$, and the network cumulative constraint violation bound is reduced to $\mathcal{O}\big( {\sqrt {\log ( T )T} } \big)$ and $\mathcal{O}\big( {\log ( T )} \big)$ without and with Slater’s condition, respectively. To the best of our knowledge, this paper is among the first to establish reduced (network) cumulative constraint violation bounds for (distributed) bandit convex optimization with time-varying constraints under Slater’s condition.
Finally, a numerical example is provided to verify the theoretical results.
\end{abstract}



\section{INTRODUCTION}
Bandit convex optimization is a sequential decision process in dynamic environments, which can be understood as a structured repeated game with $T$ iterations between a decision maker and an adversary \cite{Hazan2016a}.
Specifically, at iteration $t$, a decision maker chooses $x_t$ from a convex set $\mathbb{X} \in {\mathbb{R}^p}$ in an Euclidean space.
After committing to this choice, she receives one-point bandit feedback for a convex loss function ${f_t}:\mathbb{X} \to \mathbb{R}$ from the adversary (i.e., the value of the function ${f_t}$ at $x_t$ is revealed to the decision maker by the adversary), where $\mathbb{R}$ denotes the set of all real numbers.
Accordingly, the decision maker suffers a loss ${f_t}( {{x_t}} )$.
The goal of the decision maker is to minimize the accumulative loss $\sum\nolimits_{t = 1}^T {{f_t}( {{x_t}} )}$ over $T$ iterations. The standard measure metric is the regret
\begin{flalign}
\nonumber
\sum\limits_{t = 1}^T {{f_t}( {{x_t}} )}  - \mathop {\min }\limits_{_{x \in \mathbb{X}}} \sum\limits_{t = 1}^T {{f_t}( x )},
\end{flalign}
which measures the difference between the accumulative loss and the loss induced by the best fixed decision in hindsight.
Over the past decades, bandit convex optimization has garnered substantial interest, see, e.g., \cite{Li2023} and references therein, due to its wide applications including smart grids with uncertain supply of renewable energy \cite{Lu2013, Zhang2018a} and data centers with uncertain user demands \cite{Lin2012, Liu2013, Liu2014}.

Various bandit online algorithms with sublinear regret have been developed. For example, in~\cite{Flaxman2005}, the authors propose a bandit online projection gradient descent where the gradient of the loss function is approximated by using a one-point estimate, establishing an ${\cal O}({T^{3/4}})$ regret bound for general convex loss functions. The bound is further reduced to $\mathcal{O}( {{T^{2/3}}} )$ for more stringent strong convex loss functions in \cite{Saha2011}.
However, a lower bound from \cite{Dani2008} implies that the regret of this algorithm in \cite{Flaxman2005} will be $\Omega ( {\sqrt T } )$, even for strongly convex loss functions. The lower bound is much worse than the $\mathcal{O}\big( {\log ( T )} \big)$ regret bound established by the algorithm in \cite{Hazan2007} for the full-information feedback setting (the convex loss function is revealed to the decision maker at each iteration).
To deal with this challenge, the authors of \cite{Agarwal2010} extend bandit convex optimization by allowing that the values of the convex loss function at multiple points are simultaneously revealed to the decision maker.
Moreover, the authors propose a bandit online projection gradient descent where the gradient of the loss function is approximated by using two-point estimate, and establish an $\tilde {\mathcal{O}}( {\sqrt T } )$ regret bound and an $\mathcal{O}\big( {\log ( T )} \big)$ regret bound for general convex loss functions and strong convex loss functions, respectively.
These bounds closely resemble the optimal bounds (i.e., the $\mathcal{O}( {\sqrt T } )$ regret bound established by the algorithm in \cite{Zinkevich2003} for general convex loss functions and the $\mathcal{O}\big( {\log ( T )} \big)$ regret bound  established by the algorithm in \cite{Hazan2007} for strong convex loss functions) for the full-information feedback setting.

Note that the aforementioned algorithms require the projection operator onto the feasible set at each iteration. The operator would be straightforward if the feasible set is a simple set, e.g., a cube, a ball, or a simplex, while it would yield heavy computation burden if the feasible set is complicated. In practice, the feasible set is often characterized by inequality constraints, i.e.,
\begin{flalign}
\nonumber
\mathcal{X} = \{ {x:g( x ) \le {\mathbf{0}_m},x \in \mathbb{X}} \},
\end{flalign}
where $g( x ):{\mathbb{R}^p} \to {\mathbb{R}^m}$ is the static convex constraint function, $m$ is a positive integer, and $\mathbb{X}$ is normally a simple set. In this case, long term constraints are considered in \cite{Mahdavi2012} for bandit convex optimization, where decisions are chosen from the simple set $\mathbb{X}$ and static inequality constraints should be satisfied in the long term on average. To measure accumulative violation of inequality constraints, a performance metric, constraint violation, is defined as
\begin{flalign}
\Big\| {{{\Big[ {\sum\limits_{t = 1}^T {g( {{x_t}} )} } \Big]_ +} }} \Big\|, \label{intro-eq1}
\end{flalign}
where $\|  \cdot  \|$ denotes the Euclidean norm for vectors, $[  \cdot  ]_ + $ is the projection onto the nonnegative space. For such bandit convex optimization with long term constraints, the goal of the decision maker is to minimize regret and constraint violation. It is worth mentioning that the authors of \cite{Mahdavi2012} establish an $\mathcal{O}( {\sqrt T } )$ regret bound and an ${\cal O}({T^{3/4}})$ constraint violation bound for general convex loss functions. In \cite{Cao2019}, this problem is further extended to the time-varying constraints setting where the inequality constraint function is time-varying, and its values at several random points are revealed to the decision maker along with the values of the loss function after choosing its decision at each iteration. Moreover, the same bounds as those in \cite{Mahdavi2012} are established for general convex loss functions. 

Distributed paradigm presents a promising framework for overcoming the limitations of centralized ones including single point of the failure, data privacy and heavy computation overhead \cite{Nedic2018, Yang2019}. Therefore, distributed bandit convex optimization is increasingly garnering significant attention, see \cite{Yuan2021, Wang2020, Cao2021, Li2021a, Tu2022, Patel2022, Xiong2023}. In this problem, the loss function ${f_t}( x )$ at each iteration $t$ is decomposed across a network of $n$ agents by ${f_t}( x ) = \frac{1}{n}\sum\nolimits_{i = 1}^n {{f_{i,t}}(x)} $, where ${{f_{i,t}}}$ is called the local loss function. Each agent chooses its own decision ${x_{i,t}}$ from the set $\mathbb{X}$, and then the values of its local loss function ${{f_{i,t}}}$ at some points are revealed to itself only by the adversary. The goal of the agents is to minimize the network-wide accumulated loss, and the corresponding performance metric can be the network regret
\begin{flalign}
\nonumber
\frac{1}{n}\sum\limits_{i = 1}^n {\Big( {\sum\limits_{t = 1}^T {{f_t}( {{x_{i,t}}} )}  - \mathop {\min }\limits_{x \in X} \sum\limits_{t = 1}^T {{f_t}( {{x}} )} } \Big)}.
\end{flalign}
Recently, the authors of \cite{Yuan2021b} extend this problem by using the idea of long term constraints, and use a new form of constraint violation metric proposed in \cite{Yuan2018}. This metric is called cumulative constraint violation, and is given by
\begin{flalign}
\Big\| {\sum\limits_{t = 1}^T {{{[ {g( {{x_t}} )} ]}_ + }} } \Big\|. \label{intro-eq2}
\end{flalign}
The constraint violation defined in \eqref{intro-eq1} takes the summation across iterations before the projection operation ${[  \cdot  ]_ + }$ such that it allows strict feasible decisions having large margins compensate constraint violations at many iterations, cumulative constraint violation defined in \eqref{intro-eq2} considers all constraints that are not satisfied, and thus cumulative constraint violation is stricter than constraint violation.
To calculate the network-wide accumulated cumulative constraint violation, the corresponding performance metric can be the network cumulative constraint violation
\begin{flalign}
\nonumber
\frac{1}{n}\sum\limits_{i = 1}^n {\sum\limits_{t = 1}^T {\| {{{[ {{g}( {{x_{i,t}}} )} ]}_ + }} \|} }. 
\end{flalign}
When local loss functions are quadratic and constraint functions are linear, the authors of \cite{Yuan2021b} establish an $\mathcal{O}\big( {{T^{\max \{ {c, 1 - c} \}}}} \big)$ network regret bound and an $\mathcal{O}( {{T^{1 - c/2}}} )$ network cumulative constraint violation bound with $c \in ( {0,1} )$.
For more general convex local loss and constraint functions, the authors of \cite{Yuan2022} propose a distributed online primal--dual algorithm with one-point bandit feedback, and establish an $\mathcal{O}\big( {{T^{\max \{ {c, 1 - c/3} \}}}} \big)$ network regret bound and an $\mathcal{O}( {{T^{1 - c/2}}} )$ network cumulative constraint violation bound with $c \in ( {0,1} )$. When local loss functions are strongly convex, an $\mathcal{O}\big( {{T^{2/3}}\log ( T )} \big)$ network regret bound and an $\mathcal{O}\big( {\sqrt {T\log ( T )} } \big)$ network cumulative constraint violation bound are established.
The authors of \cite{Yi2023} further extend this problem to the time-varying constraints setting, where the time-varying convex constraint function is denoted by $g_t( x):{\mathbb{R}^p} \to {\mathbb{R}^m}$. Moreover, they propose a distributed online primal--dual algorithm with two-point bandit feedback, and establish a reduced $\mathcal{O}\big( {{T^{\max \{ {c, 1 - c} \}}}} \big)$ network regret bound and an $\mathcal{O}( {{T^{1 - c/2}}} )$ network cumulative constraint violation bound for general convex loss functions, and a reduced $\mathcal{O}( {{T^c}} )$ network regret bound and an $\mathcal{O}( {{T^{1 - c/2}}} )$ for strong convex loss functions, where $c \in ( {0,1} )$.
For calculating (network) cumulative constraint violation, a key idea in \cite{Yuan2021b, Yuan2018, Yuan2022, Yi2023} is to use the clipped constraint function ${[ g ]_ + }$ or ${[ {{g_t}} ]_ + }$ to replace the original constraint function $g$ or ${{g_t}}$, respectively.
However, in this way, reduced network cumulative constraint violation bounds cannot be established under Slater’s condition \cite{Yi2024}, i.e., the use of the clipping operation renders Slater’s condition ineffective. 
Slater’s condition is a sufficient condition for strong duality to hold in convex optimization problems \cite{Boyd2004}, which is used to establish reduced constraint violation for the full-information feedback setting in \cite{Yu2017, Neely2017}.
Moreover, the authors of \cite{Yi2024} propose a distributed online primal--dual algorithm with full-information feedback where the clipped constraint function ${[ {{g_t}} ]_ + }$ is not directly used to replace the original constraint function ${{g_t}}$. Instead, the algorithm updates the dual variables by directly maximizing the regularized Lagrangian function. In particular, with this idea, Slater’s condition is still effective.

Note that such a distributed bandit online algorithm that can achieve reduced network cumulative constraint violation under Slater’s condition is still missing, which motivates our study. In particular, in this paper, we consider the distributed bandit convex optimization problem with time-varying constraints where the decision makers receive bandit feedback for both loss and constraint functions at each iteration, and use network regret and cumulative constraint violation as performance metrics.
The contributions are summarized as follows.\vspace{-1pt}
\begin{itemize}
\item[$\bullet$]
This paper proposes a new distributed bandit online primal--dual algorithm. Different from the distributed bandit online algorithm in \cite{Yi2023} where the clipped constraint function ${[ {{g_t}} ]_ + }$ is directly used to replace the original constraint function ${{g_t}}$, the proposed algorithm updates the dual variables by directly maximizing the regularized Lagrangian function, which can be explicitly calculated using the clipped constraint function. Different from the distributed online algorithm in \cite{Yi2024} where the dual variables are updated by using the composite objective mirror descent and the subgradients of local loss and constraint functions are directly used, the proposed algorithm updates the dual variables by using the projected gradient descent instead of the composite objective mirror descent used in \cite{Yi2024}. Moreover, the proposed algorithm uses two-point stochastic estimators to approximate these subgradients as they are unavailable in the bandit feedback setting.
Note that the gaps between the estimators and the subgradients cause nontrivial challenges for performance analysis, which will be explained in detail in Remark~2. More importantly, the proposed algorithm enables the use of the stricter cumulative constraint violation metric while preserving the effectiveness of Slater’s condition.

\item[$\bullet$]
For convex local loss functions, we show in Theorem~1 that the proposed algorithm establishes an $\mathcal{O}\big( {{T^{\max \{ {c,1-c} \}}}} \big)$ network regret bound and an $\mathcal{O}( {{T^{1 - c/2}}} )$ network cumulative constraint violation bound with $c \in ( {0,1} )$, which are the same as the results established in \cite{Yi2023, Yi2024}, generalize the results established in \cite{Mahdavi2012, Cao2019, Yuan2021b}, and also improve the results established in \cite{Yuan2022}. When Slater’s condition holds, we show in Theorem~2 that the network cumulative constraint violation bound is reduced to $\mathcal{O}( {{T^{1 - c}}} )$. To the best of our knowledge, this paper is among the first to establish a reduced cumulative constraint violation bound for bandit convex optimization with long term constraints under Slater’s condition.
\item[$\bullet$]
For strongly convex local loss functions, under unknown strongly convex parameters, we show in Theorem~3 that the proposed algorithm establishes an $\mathcal{O}( {{T^{1 - c}}} )$ network regret bound and an $\mathcal{O}( {{T^{1 - c/2}}} )$ network cumulative constraint violation bound with $c \in ( {0,1} )$, moreover, the network cumulative constraint violation bound is reduced to $\mathcal{O}( {{T^{1 - c}}} )$ when Slater’s condition holds. Under known strongly convex parameters, we show in Theorem~4 that the proposed algorithm establishes an $\mathcal{O}\big( {\log ( T )} \big)$ network regret bound and an $\mathcal{O}\big( {\sqrt {\log ( T )T} } \big)$ network cumulative constraint violation bound, which are the same as the results established in \cite{Yi2024}, and improve the results established in \cite{Yuan2022, Yi2023}. Moreover, the network cumulative constraint violation bound is reduced to $\mathcal{O}\big( {\log ( T )} \big)$ when Slater’s condition holds. Note that this paper is among the first to establish such a result for (distributed) bandit convex optimization with long term constraints.
\end{itemize}

\begin{table*}
\centering
\caption{Comparison of this paper to related works on bandit convex optimization \\with long term constraints.}
\resizebox{\linewidth}{!}{
\begin{tabular}{c|c|c|c|c|c|c|c|c}
\Xcline{1-9}{1pt}
\multicolumn{2}{c|}{\multirow{3}{*}{Reference}} & {\multirow{3}{*}{\makecell{Problem \\type}}} & \multirow{3}{*}{\makecell{Loss and \\constraint \\functions}} & \multirow{3}{*}{\makecell{Slater's \\condition}} & \multirow{3}{*}{\makecell{Information \\feedback}} & \multirow{3}{*}{Regret} & \multirow{3}{*}{\makecell{Constraint \\violation}} & \multirow{3}{*}{\makecell{Cumulative \\constraint \\violation}}\\
\multicolumn{2}{c|}{} & & & & & & &\\
\multicolumn{2}{c|}{} & & & & & & &\\
\cline{1-9}
\multicolumn{2}{c|}{\multirow{3}{*}{\cite{Mahdavi2012}}} & {\multirow{3}{*}{Centralized}} & \multirow{3}{*}{\makecell{Convex, \\convex and \\static}} & \multirow{3}{*}{\makecell{$\times$}} & \multirow{3}{*}{\makecell{$\nabla {f_t}$, and two-point \\bandit feedback \\for  ${g}$}} & \multirow{3}{*}{$\mathcal{O}( {\sqrt T } )$} & \multirow{3}{*}{$\mathcal{O}( {{T^{3/4}}} )$} & \multirow{3}{*}{Not given}\\
\multicolumn{2}{c|}{} & & & &  & &\\
\multicolumn{2}{c|}{} & & & & & &\\
\cline{1-9}
\multicolumn{2}{c|}{\multirow{3}{*}{\cite{Cao2019}}} & {\multirow{3}{*}{Centralized}} & \multirow{3}{*}{\makecell{Convex, \\convex and \\time-varying}} & \multirow{3}{*}{\makecell{$\times$}} & \multirow{3}{*}{\makecell{Two-point \\bandit feedback \\for ${f_t}$ and ${g_t}$}} & \multirow{3}{*}{$\mathcal{O}( {\sqrt T } )$} & \multirow{3}{*}{$\mathcal{O}( {{T^{3/4}}} )$} & \multirow{3}{*}{Not given}\\
\multicolumn{2}{c|}{} & & & & & &\\
\multicolumn{2}{c|}{} & & & & & &\\
\cline{1-9}
\multicolumn{2}{c|}{\multirow{3}{*}{\cite{Yuan2021b}}} & {\multirow{3}{*}{Distributed}} & \multirow{3}{*}{\makecell{Quadratic, \\linear and \\static}} & \multirow{3}{*}{\makecell{$\times$}} & \multirow{3}{*}{\makecell{$\nabla {g}$, and two-point \\bandit feedback \\for ${f_t}$}} & \multirow{3}{*}{$\mathcal{O}\big( {{T^{\max \{ {c, 1 - c} \}}}} \big)$} & \multicolumn{2}{c}{\multirow{3}{*}{$\mathcal{O}( {{T^{1 - c/2}}} )$}} \\
\multicolumn{2}{c|}{} & & & & & \\
\multicolumn{2}{c|}{} & & & & & \\
\cline{1-9}
\multicolumn{2}{c|}{\multirow{6}{*}{\cite{Yuan2022}}} & {\multirow{6}{*}{Distributed}} & \multirow{3}{*}{\makecell{Convex, \\convex and \\static}} & \multirow{6}{*}{\makecell{$\times$}} & \multirow{6}{*}{\makecell{$\nabla {g}$, and one-point \\bandit feedback \\for ${f_t}$}} & \multirow{3}{*}{$\mathcal{O}\big( {{T^{\max \{ {c, 1 - c/3} \}}}} \big)$} & \multicolumn{2}{c}{\multirow{3}{*}{$\mathcal{O}( {{T^{1 - c/2}}} )$}}\\
\multicolumn{2}{c|}{} & & & & & \\
\multicolumn{2}{c|}{} & & & & & \\
\cline{4-4}
\cline{7-9}
\multicolumn{2}{c|}{} & & \multirow{3}{*}{\makecell{Strongly convex, \\convex and \\static}} & & & \multirow{3}{*}{$\mathcal{O}\big( {{T^{2/3}}\log ( T )} \big)$} & \multicolumn{2}{c}{\multirow{3}{*}{$\mathcal{O}\big( {\sqrt {T\log ( T )} } \big)$}} \\
\multicolumn{2}{c|}{} & & & & & \\
\multicolumn{2}{c|}{} & & & & & \\
\cline{1-9}
\multicolumn{2}{c|}{\multirow{6}{*}{\cite{Yi2023}}} & {\multirow{6}{*}{Distributed}} & \multirow{3}{*}{\makecell{Convex, \\convex and \\time-varying}} & \multirow{6}{*}{\makecell{$\times$}} & \multirow{6}{*}{\makecell{Two-point \\bandit feedback \\for ${f_t}$ and ${g_t}$}} & \multirow{3}{*}{$\mathcal{O}\big( {{T^{\max \{ {c, 1 - c} \}}}} \big)$} & \multicolumn{2}{c}{\multirow{3}{*}{$\mathcal{O}( {{T^{1 - c/2}}} )$}}\\
\multicolumn{2}{c|}{} & & & & & \\
\multicolumn{2}{c|}{} & & & & & \\
\cline{4-4}
\cline{7-9}
\multicolumn{2}{c|}{} & & \multirow{3}{*}{\makecell{Strongly convex, \\convex and \\time-varying}} & & & \multirow{3}{*}{$\mathcal{O}( {{T^c}} )$} & \multicolumn{2}{c}{\multirow{3}{*}{$\mathcal{O}( {{T^{1 - c/2}}} )$}}\\
\multicolumn{2}{c|}{} & & & & & \\
\multicolumn{2}{c|}{} & & & & & \\
\cline{1-9}
\multicolumn{2}{c|}{\multirow{8}{*}{\makecell{This \\paper}}} & {\multirow{8}{*}{Distributed}} & \multirow{4}{*}{\makecell{Convex, \\convex and \\time-varying}} & \multirow{2}{*}{\makecell{$\times$}} & \multirow{8}{*}{\makecell{Two-point \\bandit feedback \\for ${f_t}$ and ${g_t}$}} & \multirow{4}{*}{$\mathcal{O}\big( {{T^{\max \{ {c, 1 - c} \}}}} \big)$} & \multicolumn{2}{c}{\multirow{2}{*}{$\mathcal{O}( {{T^{1 - c/2}}} )$}}\\
\multicolumn{2}{c|}{} & & & & & \\
\cline{5-5}
\cline{8-9}
\multicolumn{2}{c|}{} & & & \multirow{2}{*}{\makecell{$\checkmark$}} & & & \multicolumn{2}{c}{\multirow{2}{*}{$\mathcal{O}( {{T^{1 - c}}} )$}} \\
\multicolumn{2}{c|}{} & & & & & \\
\cline{4-5}
\cline{7-9}
\multicolumn{2}{c|}{} & & \multirow{4}{*}{\makecell{Strongly convex, \\convex and \\time-varying}} & \multirow{2}{*}{\makecell{$\times$}} & & \multirow{4}{*}{$\mathcal{O}\big( {\log ( T )} \big)$} & \multicolumn{2}{c}{\multirow{2}{*}{$\mathcal{O}\big( {\sqrt {\log ( T )T} } \big)$}}\\
\multicolumn{2}{c|}{} & & & & & \\
\cline{5-5}
\cline{8-9}
\multicolumn{2}{c|}{} & & & \multirow{2}{*}{\makecell{$\checkmark$}} & & & \multicolumn{2}{c}{\multirow{2}{*}{$\mathcal{O}\big( {\log ( T )} \big)$}}\\
\multicolumn{2}{c|}{} & & & & & \\
\Xcline{1-9}{1pt}
\end{tabular}
}
\end{table*}

The detailed comparison of this paper to related studies is summarized in TABLE~I, where we only present the static part of regret for the sake of clarity.

The remainder of this paper is organised as follows.
Section~II presents the problem formulation.
Section~III proposes the distributed bandit online primal--dual algorithm, and analyzes its performance.
Section~IV demonstrates a numerical simulation to verify the theoretical results.
Finally, Section~V concludes this paper.
All proofs are given in Appendix.

\textbf{Notations:} ${\mathbb{N}_ + }$, $\mathbb{R}$, ${\mathbb{R}^p}$ and $\mathbb{R}_ + ^p$ denote the sets of all positive integers, real numbers, $p$-dimensional and nonnegative vectors, respectively. Given $m$ and $n \in {\mathbb{N}_ + }$, $[ m ]$ denotes the set $\{ {1, \cdot  \cdot  \cdot ,m} \}$, and $[m, n]$ denotes the set $\{ {m, \cdot  \cdot  \cdot ,n} \}$ for $m < n$. Given vectors $x$ and $y$, ${x^T}$ denotes the transpose of the vector $x$, and $\langle {x,y} \rangle $ and $x \otimes y$ denote the standard inner and Kronecker product of the vectors $x$ and $y$, respectively. ${\mathbf{0}_m}$ denotes the $m$-dimensional column vector whose components are all $0$. $\mathrm{col}( {q_1}, \cdot  \cdot  \cdot ,{q_n} )$ denotes the concatenated column vector of ${q_i} \in {\mathbb{R}^{{m_i}}}$ for $i \in [ n ]$. ${\mathbb{B}^p}$ and ${\mathbb{S}^p}$ denote the unit ball and sphere centered around the origin in ${\mathbb{R}^p}$ under Euclidean norm, respectively. $\mathbf{E}$ denotes the expectation. For a set $\mathbb{K} \in {\mathbb{R}^p}$ and a vector $ x \in {\mathbb{R}^p}$, ${\mathcal{P}_{\mathbb{K}}}(  x  )$ denotes the projection of the vector $x$ onto the set $\mathbb{K}$, i.e., ${\mathcal{P}_{\mathbb{K}}}( x ) = \arg {\min _{y \in {\mathbb{K}}}}{\| {x - y} \|^2}$, and $[  x  ]_+$ denotes ${\mathcal{P}_{\mathbb{R}_ + ^p}}( x )$. For a function $f$ and a vector $ x $, $\partial f( x )$ denotes the subgradient of $f$ at $x$.

\section{PROBLEM FORMULATION}
Consider the distributed bandit convex optimization problem with time-varying constraints.
At iteration $t$, a network of $n$ agents is modeled by a time-varying directed graph ${\mathcal{G}_t} = ( {\mathcal{V},{\mathcal{E}_t}} )$ with the agent set $\mathcal{V} = [ n ]$ and the edge set ${\mathcal{E}_t} \subseteq \mathcal{V} \times \mathcal{V}$. $( {j,i} ) \in {\mathcal{E}_t}$ indicates that agent $i$ can receive information from agent $j$.
The sets of in- and out-neighbors of agent $i$ are $\mathcal{N}_i^{\text{in}}( {{\mathcal{G}_t}} ) = \{ {j \in [ n ]|( {j,i} ) \in {\mathcal{E}_t}} \}$ and $\mathcal{N}_i^{\text{out}}( {{\mathcal{G}_t}} ) = \{ {j \in [ n ]|( {i,j} ) \in {\mathcal{E}_t}} \}$, respectively.
An adversary first erratically selects $n$ convex functions $\{ {{f_{i,t}}:\mathbb{X} \to \mathbb{R}} \}$ and $n$ convex constraint functions $\{ {{g_{i,t}}:\mathbb{X} \to {\mathbb{R}^{{m_i}}}} \}$ for $i \in [n]$, where $\mathbb{X} \subseteq {\mathbb{R}^p}$ is a known convex set, and both ${m_i}$ and $p$ are positive integers. Then, the agents collaborate to select their local decisions $\{ {{x_{i,t}} \in \mathbb{X}} \}$ without prior access to $\{ {{f_{i,t}}} \}$ and $\{ {{g_{i,t}}} \}$. At the same time, the values of ${f_{i,t}}$ and ${g_{i,t}}$ at the point ${x_{i,t}}$ as well as at other potential points are privately revealed to each agent $i$.
The goal of the agents is to choose the decision sequence $\{ {{x_{i,t}}} \}$ for $i \in [n]$ and $t \in [T]$ such that both network regret
\begin{flalign}
{\rm{Net}\mbox{-}\rm{Reg}}( T ) &:= \frac{1}{n}\sum\limits_{i = 1}^n {\Big( {\sum\limits_{t = 1}^T {{f_t}( {{x_{i,t}}} )}  - \mathop {\min }\limits_{x \in {\mathcal{X}_T}} \sum\limits_{t = 1}^T {{f_t}( x )} } \Big)}, \label{regret-eq1}
\end{flalign}
and network cumulative constraint violation
\begin{flalign}
{\rm{Net}\mbox{-}\rm{CCV}}( T ) &:= \frac{1}{n}\sum\limits_{i = 1}^n {\sum\limits_{t = 1}^T {\| {{{[ {{g_t}( {{x_{i,t}}} )} ]}_ + }} \|} }, \label{CCV-eq2}
\end{flalign}
increase sublinearly, where ${f_t}( x ) = \frac{1}{n}\sum\nolimits_{j = 1}^n {{f_{j,t}}( x )} $ is the global loss function of the network at iteration $t$, ${g_t}( x ) = {\rm{col}}\big( {{g_{1,t}}( x ), \cdot  \cdot  \cdot ,{g_{n,t}}( x )} \big) \in {\mathbb{R}^m}$ with $m = \sum\nolimits_{i = 1}^n {{m_i}} $ is the global constraint function of the network at iteration $t$, and
\begin{flalign}
{\mathcal{X}_T} = \{ { x :x \in \mathbb{X}, {g_t}( x ) \le {\mathbf{0}_m},\forall t \in [ T ]} \}, \label{pl-eq6}
\end{flalign}
is the feasible set. To guarantee that the offline optimal static decision always exists, we assume that for any $T \in {\mathbb{N}_ + }$, the feasible set ${\mathcal{X}_T}$ is nonempty.

Note that when ${g_{i,t}} \equiv {\mathbf{0}_{{m_i}}}$, $\forall i \in [ n ]$, $\forall t \in {\mathbb{N}_ + }$, the considered distributed problem becomes the problem studied in \cite{Wang2020, Cao2021, Li2021a, Tu2022, Patel2022, Xiong2023}; when ${g_{i,t}} \equiv g$, $\forall i \in [ n ]$, $\forall t \in {\mathbb{N}_ + }$ with $g$ being a known constraint function, the considered distributed problem becomes the problem studied in \cite{Yuan2021b, Yuan2022}.

The following commonly used assumption on the set $\mathbb{X}$ is made.
\begin{assumption}
The set $\mathbb{X}$ is closed.
Moreover, the convex set $\mathbb{X}$ contains the ball of radius $r( \mathbb{X} )$ and is contained in the ball of radius $R( \mathbb{X} )$, i.e.,
\begin{flalign}
r( \mathbb{X} ){\mathbb{B}^p} \subseteq \mathbb{X} \subseteq R( \mathbb{X} ){\mathbb{B}^p}. \label{ass1-eq1}
\end{flalign}
\end{assumption}

The following assumptions on the loss and constraint functions are made.
\begin{assumption}
For all $i \in [n]$, $t \in {\mathbb{N}_ + }$, there exists a constant ${F}$ such that
\begin{flalign}
| {{f_{i,t}}( x )} | &\le {F}. \label{ass2-eq1a}
\end{flalign}
\end{assumption}

\begin{assumption}
For all $i \in [n]$, $t \in {\mathbb{N}_ + }$, the functions ${f_{i,t}}$ and ${{g_{i,t}}}$ are convex, and the subgradients $\partial {f_{i,t}}( x )$ and $\partial {g_{i,t}}( x )$ exist. Moreover, there exist constants ${G_1}$ and ${G_2}$ such that
\begin{subequations}
\begin{flalign}
\| {\partial {f_{i,t}}( x )} \| &\le {G_1}, \label{ass4-eq1a}\\
\| {\partial {g_{i,t}}( x )} \| &\le {G_2}, x \in \mathbb{X}. \label{ass4-eq1c}
\end{flalign}
\end{subequations}
\end{assumption}

Note that we do not need the assumption that the local constraint functions $\{ {{g_{i,t}}} \}$ are uniformly bounded while \cite{Yi2023} needs it.
In addition, from Assumption~3, and Lemma 2.6 in \cite{ShalevShwartz2012}, for all $i \in [n]$, $t \in {\mathbb{N}_ + }$, we have
\begin{subequations}
\begin{flalign}
| {{f_{i,t}}( x ) - {f_{i,t}}( y )} | &\le {G_1}\| {x - y} \|, \label{ass4-eq2a}\\
\| {{g_{i,t}}( x ) - {g_{i,t}}( y )} \| &\le {G_2}\| {x - y} \|, x,y \in \mathbb{X}. \label{ass4-eq2c}
\end{flalign}
\end{subequations}

The following assumption on the graph is made, which is also used in \cite{Yi2020, Yi2023, Yi2024, Yi2021b}.
\begin{assumption}
For $t \in {\mathbb{N}_ + }$, the time-varying directed graph $\mathcal{G}_t$ satisfies that

\noindent (i) There exists a constant $w  \in ( {0,1} )$ such that ${[ {{W_t}} ]_{ij}} \ge w$ if $( {j,i} ) \in {\mathcal{E}_t}$ or $i = j$, and ${[ {{W_t}} ]_{ij}} = 0$ otherwise.

\noindent (ii) The mixing matrix ${W_t}$ is doubly stochastic, i.e., ${\sum\nolimits_{i = 1}^n {[ {{W_t}} ]} _{ij}} = {\sum\nolimits_{j = 1}^n {[ {{W_t}} ]} _{ij}} = 1$, $\forall i,j \in [ n ]$.

\noindent (iii) There exists an integer $B > 0$ such that the time-varying directed graph $( {\mathcal{V}, \cup _{l = 0}^{B - 1}{\mathcal{E} _{t + l}}} )$ is strongly connected.
\end{assumption}

As stated in the introduction, compared to \cite{Yi2023}, this paper establishes reduced network cumulative constraint violation bounds under Slater’s condition. In the following, we formally introduce Slater’s condition.
\begin{assumption}
(Slater’s condition)
There exists a point ${x_s} \in \mathbb{X}$ and a positive constant ${\varsigma _s}$ such that
\begin{flalign}
{g_t}( {{x_s}} ) \le  - {\varsigma _s}{\mathbf{1}_m},t \in {\mathbb{N}_ + }. \label{ass8-eq1}
\end{flalign}
\end{assumption}
Slater’s condition is a sufficient condition for strong duality to hold in convex optimization problems \cite{Boyd2004}. To the best of our knowledge, there are no studies to show that reduced cumulative constraint violation bounds can be established in bandit convex optimization. To calculate network cumulative constraint violation, \cite{Yi2023} directly replaces the original constraint functions with the corresponding clipped constraint functions, which makes Slater’s condition ineffective. In this paper, we propose a new distributed bandit online algorithm where Slater’s condition remains effective.

\section{DISTRIBUTED BANDIT ONLINE PRIMAL--DUAL ALGORITHM}

\subsection{Algorithm description}
For the global loss function ${f_t}$ and constraint function ${g_t}$, the associated regularized Lagrangian function is
\begin{flalign}
{\mathcal{L}_t}( {{x_t},{q_t}} ): = \frac{1}{n}\sum\limits_{i = 1}^n {{f_{i,t}}( {{x_t}} )} + q_t^T{g_t}( {{x_t}} ) - \frac{1}{{2{\gamma _t}}}{\| {{q_t}} \|^2}, \label{alg-eq1}
\end{flalign}
where ${x_t} \in {\mathbb{R}^p}$ and ${q_t} \in \mathbb{R}_ + ^m$ represent the primal and dual variables, respectively, and ${{\gamma _t}}$ is the regularization parameter.
The primal and dual variables can be updated by the standard projected primal--dual algorithm
\begin{subequations}
\begin{flalign}
{q_{t + 1}} &= {\Big[ {{q_t} + {\alpha _t}\frac{{\partial {\mathcal{L}_t}( {{x_t},q} )}}{{\partial q}}\Big| {_{q = {q_t}}} \Big.} \Big]_ + }, \label{alg-eq2} \\
{x_{t + 1}} &= {\mathcal{P}_\mathbb{X}}\Big( {{x_t} - {\alpha _t}\frac{{\partial {\mathcal{L}_t}( {x,{q_{t + 1}}} )}}{{\partial x}}\Big| {_{x = {x_t}}}} \Big), \label{alg-eq3}
\end{flalign}
\end{subequations}
where ${{\alpha _t}}$ is the stepsize. To calculate cumulative constraint violation, the clipped constraint function ${[ {{g_t}( {{x_t}} )} ]_ + }$ is directly used to replace the original constraint function ${{g_t}( {{x_t}} )}$ in \eqref{alg-eq1} in \cite{Yi2023}.
However, that makes Slater’s condition ineffective. To deal with this dilemma, we still use the original constraint function ${{g_t}( {{x_t}} )}$ in \eqref{alg-eq1} but maximize ${{\mathcal{L}_t}( {{x_t},q} )}$ over all $q \in \mathbb{R}_ + ^m$ to replace \eqref{alg-eq2}, i.e.,
\begin{flalign}
{q_{t + 1}} = \mathop {\arg \max }\limits_{q \in \mathbb{R}_ + ^m} {\mathcal{L}_t}( {{x_t},q} ) = {\gamma _t}{[ {{g_t}( {{x_t}} )} ]_ + }, \label{alg-eq4}
\end{flalign}
where the second equation holds due to $- 2{\gamma _t}{g_t}( {{x_t}} )/ - 2 = {\gamma _t}{g_t}( {{x_t}} )$. The updating rule \eqref{alg-eq4} is also adopted in \cite{Yuan2018, Yuan2022, Yi2024}.

To implement the updating rules in a distributed manner, we use ${x_{i,t}}$ to denote the local copy of the primal variable ${x_t}$, and rewrite the dual variable in an agent-wise manner, i.e., ${q_t} = {\rm{col}}( {{q_{1,t}}, \cdot  \cdot  \cdot ,{q_{n,t}}} )$ with each ${q_{1,t}} \in \mathbb{R}_ + ^{{m_i}}$.
Then, the updating rule \eqref{alg-eq4} can be executed in an agent-wise manner as \eqref{Algorithm1-eq2}. Note that in the bandit setting, the subgradients are unavailable and only the values of ${f_{i,t}}$ and ${g_{i,t}}$ at some potential points are privately revealed. Thus, we use the values of the local loss function ${f_{i,t}}$ at ${x_{i,t}}$ and ${x_{i,t}} + {\delta _t}{u_{i,t}}$ to estimate the subgradient $\partial {f_{i,t}}( {{x_{i,t}}} )$, and use the values of the local constraint function ${g_{i,t}}$ at ${x_{i,t}}$ and ${x_{i,t}} + {\delta _t}{u_{i,t}}$ to estimate the subgradient $\partial{{g_{i,t}}( {{x_{i,t}}} )}$, i.e.,
\begin{flalign}
\nonumber
\hat \partial {f_{i,t}}( {{x_{i,t}}} ) &= \frac{p}{{{\delta _t}}}\big( {{f_{i,t}}( {{x_{i,t}} + {\delta _t}{u_{i,t}}} ) - {f_{i,t}}( {{x_{i,t}}} )} \big){u_{i,t}} \in {\mathbb{R}^p}, \\
\nonumber
\hat \partial {{g_{i,t}}( {{x_{i,t}}} )} &= \frac{p}{{{\delta _t}}}{\big( {{{{g_{i,t}}( {{x_{i,t}} + {\delta _t}{u_{i,t}}} )} } - { {{g_{i,t}}( {{x_{i,t}}} )} }} \big)^T}
\otimes {u_{i,t}} \in {\mathbb{R}^{p \times {m_i}}},
\end{flalign}
where ${\delta _t} \in ( {0,r( \mathbb{X} ){\xi _t}} ]$ is an exploration parameter, ${r( \mathbb{X} )}$ is a positive constant, ${\xi _t} \in ( {0,1} )$ is a shrinkage coefficient, and ${u_{i,t}} \in {\mathbb{S}^p}$ is a uniformly distributed random vector. The idea follows the two-point stochastic subgradient estimator proposed in \cite{Agarwal2010, Shamir2017}, and is also adopted in \cite{Yi2023, Yi2021b}. ${{\hat \omega }_{i, t + 1}}$ defined in \eqref{Algorithm1-eq3} can be understood as an estimator for a portion of $\frac{{\partial {\mathcal{L}_t}( {x,{q_{t + 1}}} )}}{{\partial x}}\Big| {_{x = {x_t}}}$ that is available to agent $i$. Then, each ${z_{i,t + 1}}$ updated by \eqref{Algorithm1-eq4} can be understood as a local estimate of ${x_{t + 1}}$ updated by \eqref{alg-eq3}. To estimate ${x_{t + 1}}$ more accurately, each agent $i$ computes ${x_{i,t + 1}}$ by the consensus protocol \eqref{Algorithm1-eq1}, which tracks the average $\frac{1}{n}\sum\nolimits_{i = 1}^n {{z_{i,t + 1}}}$. As a result, the distributed bandit online primal--dual algorithm is proposed, which is presented in pseudo-code as Algorithm~1.
\begin{algorithm}[!t]
  \caption{Distributed Bandit Online Primal--Dual Algorithm} 
  \begin{algorithmic}
  \renewcommand{\algorithmicrequire}{\textbf{Input:}}
  \REQUIRE
    constant ${r( \mathbb{X} )}$, non-increasing sequences $\{ {\alpha _t}\} \subseteq ( {0, + \infty })$, $\{ {{\xi _t}} \} \subseteq ( {0,1} )$, $\{ {{\delta _t}} \} \subseteq ( {0,r( \mathbb{X} ){\xi _t}} ]$, and non-decreasing sequence $\{ {\gamma _t}\} \subseteq ( {0, + \infty })$.
  \renewcommand{\algorithmicrequire}{\textbf{Initialize:}}
  \REQUIRE
       ${z_{i,1}} \in {(1-\xi_{1})\mathbb{X}}$.
    \FOR {$t = 1, \cdot  \cdot  \cdot, T-1 $}
    \FOR {$i = 1,\cdot  \cdot  \cdot,n$ in parallel}
    \STATE Broadcast $z_{i,t}$ to $\mathcal{N}_i^{\text{out}}( {{\mathcal{G}_t}} )$ and receive $z_{j,t}$ from $j \in \mathcal{N}_i^{\text{in}}( {{\mathcal{G}_t}} )$.
    \STATE Select\par\nobreak\vspace{-10pt}
    \begin{small}
     \begin{flalign}
       {x_{i,t}} &= \sum\limits_{j = 1}^n {{{[{W_t}]}_{ij}}{z_{j,t}}}. \label{Algorithm1-eq1}
    \end{flalign}
      \end{small}%
    \STATE Select vector ${u_{i,t}} \in {\mathbb{S}^p}$ independently and uniformly at random.
    \STATE Observe ${{f_{i,t}}( {{x_{i,t}}} )}$, ${{f_{i,t}}( {{x_{i,t}} + {\delta _t}{u_{i,t}}} )}$, ${{ {{g_{i,t}}( {{x_{i,t}}} )} }}$, ${{ {{g_{i,t}}( {{x_{i,t}} + {\delta _t}{u_{i,t}}} )}  }}$.
    \STATE Update\par\nobreak\vspace{-10pt}
    \begin{subequations}
     \begin{flalign}
       {q_{i,t + 1}} &= {\gamma _t}{[ {{g_{i,t}}( {{x_{i,t}}} )} ]_ + }, \label{Algorithm1-eq2}\\
       {\hat{\omega} _{i,t + 1}} &= \hat{\partial} {f_{i,t}}({x_{i,t}}) + \hat{\partial} {{g_{i,t}}({x_{i,t}}) }{q_{i,t+1}}, \label{Algorithm1-eq3}\\
       {z_{i,t + 1}} &= {\mathcal{P}_{( {1 - {\xi _{t + 1}}} )\mathbb{X}}}( {{x_{i,t}} - {\alpha _t}{{\hat \omega }_{i,t + 1}}} ). \label{Algorithm1-eq4}
      \end{flalign}
      \end{subequations}
    \ENDFOR
    \ENDFOR
  \renewcommand{\algorithmicensure}{\textbf{Output:}}
  \ENSURE
      $\{ x_{i,t} \}$.
  \end{algorithmic}
\end{algorithm}

\subsection{Performance Analysis for Convex Case}
In this subsection, we establish network regret and cumulative constraint violation bounds for Algorithm~1 in the following theorems when local loss functions are convex. We first establish these bounds without Slater’s condition.

\begin{theorem}\label{thm1}
Suppose Assumptions 1--4 hold. For all $i \in [ n ]$, let $\{ {{x_{i,t}}} \}$ be the sequences generated by Algorithm~1 with
\begin{flalign}
{\alpha _t} = \frac{1}{{{t^c}}}, {\gamma _t} = \frac{{{\gamma _0}}}{{{\alpha _t}}}, {\xi _t} = {\alpha _t}, {\delta _t} = r( \mathbb{X} ){\alpha _t}, \label{theorem1-eq1}
\end{flalign}
where $c \in ( {0,1} )$, ${\gamma _0} \in \big( {0, 1 /( {4(p^2+1)G_2^2} )} \big]$ are constants. Then, for any $T \in {\mathbb{N}_ + }$,
\begin{flalign}
\mathbf{E}[{{\rm{Net}\mbox{-}\rm{Reg}}( T )}] &= \mathcal{O}( {T^{\max \{ {c,1 - c} \}}} ), \label{theorem1-eq2}\\
\mathbf{E}[{{\rm{Net} \mbox{-} \rm{CCV}}( T )}] &= \mathcal{O}( {T^{1 - c/2}} ). \label{theorem1-eq3}
\end{flalign}
\end{theorem}
The proof and the explicit expressions of the right-hand sides of \eqref{theorem1-eq2}--\eqref{theorem1-eq3} are given in Appendix B.
\begin{remark}\label{rem1}
In Theorem~1, we show that Algorithm~1 establishes sublinear network regret and cumulative constraint violation bounds.
The bounds \eqref{theorem1-eq2} and \eqref{theorem1-eq3} are the same as the state-of-the-art bounds established by the distributed bandit online algorithm in \cite{Yi2023}. Note that the potential drawback of Algorithm~1 is that it uses ${G_2}$ to design the algorithm parameter ${\gamma _0}$. However, we do not use the assumption that local constraint functions are uniformly bounded while \cite{Yi2023} uses it. The bounds \eqref{theorem1-eq2} and \eqref{theorem1-eq3} are also the same as the bounds established by the distributed online algorithm with full-information feedback in \cite{Yi2024}. If setting $c = 1/2$, they then generalize the results established in \cite{Mahdavi2012, Cao2019, Yuan2021b}, even though the bandit online algorithms in \cite{Mahdavi2012, Cao2019} are centralized and the more tolerant constraint violation metric is used, and in \cite{Yuan2021b} local loss functions are quadratic and local constraint functions are linear, static and known in advance. Moreover, the bounds \eqref{theorem1-eq2} and \eqref{theorem1-eq3} improve the results established by the distributed online algorithm with one-point bandit feedback in \cite{Yuan2022}, even though in \cite{Yuan2022} local constraint functions are static and known in advance.
\end{remark}

\begin{remark}\label{rem2}
Different from the algorithm in \cite{Yi2024} that directly uses the subgradients of local loss and constraint functions, Algorithm~1 is based on the two-point stochastic gradient estimators, which are unbiased subgradients of the uniformly smoothed versions of local loss and constraint functions. In the one-dimensional case ($p=1$), the intuition is readily seen that the expectations of the estimators $\hat \partial {f_{i,t}}( {{x_{i,t}}} )$ and $\hat \partial {{g_{i,t}}( {{x_{i,t}}} )}$ equal $\frac{1}{{2{\delta _t}}}({f_{i,t}}({x_{i,t}} + {\delta _t}) - {f_{i,t}}({x_{i,t}} - {\delta _t}))$ and $\frac{1}{{2{\delta _t}}}({g_{i,t}}({x_{i,t}} + {\delta _t}) - {g_{i,t}}({x_{i,t}} - {\delta _t}))$. We know that they indeed approximate the derivatives of ${f_{i,t}}$ and ${g_{i,t}}$ at ${x_{i,t}}$ if ${\delta _t}$ is infinitesimal.
However, ${\delta _t}$ cannot be small enough, and thus there exist gaps between $\hat \partial {f_{i,t}}( {{x_{i,t}}} )$, $\hat \partial {{g_{i,t}}( {{x_{i,t}}} )}$ and their true subgradients. These gaps prevent us from simply using the uniformly bounds of the subgradients in Assumption~3 as done in \cite{Yi2024} to bound the estimators, and thus cause nontrivial challenges for performance analysis. To cope with this challenge, we need to analyze some properties of the estimators such as the uniformly bounds of the estimators and the gaps between the local loss and constraint functions and the correspondingly uniformly smoothed functions. Moreover, Algorithm~1 updates the dual variables by using the projected gradient descent instead of the composite objective mirror descent used in \cite{Yi2024}. Therefore, the proof of our Theorem~1 has significant differences compared to that of Theorem~1 in \cite{Yi2024}.
\end{remark}

With Slater’s condition, we show that Algorithm~1 establishes a reduced network cumulative constraint violation bound than the bound established in \eqref{theorem1-eq3} in the following theorem.
\begin{theorem}\label{thm2}
Suppose Assumptions 1--5 hold. For all $i \in [ n ]$, let $\{ {{x_{i,t}}} \}$ be the sequences generated by Algorithm~1 with \eqref{theorem1-eq1}. Then, for any $T \in {\mathbb{N}_ + }$,
\begin{flalign}
\mathbf{E}[{{\rm{Net}\mbox{-}\rm{Reg}}( T )}] &= \mathcal{O}( {T^{\max \{ {c,1 - c} \}}} ), \label{theorem2-eq1} \\
\mathbf{E}[{{\rm{Net} \mbox{-} \rm{CCV}}( T )}] &= \mathcal{O}( {T^{1 - c}} ). \label{theorem2-eq2}
\end{flalign}
\end{theorem}
The proof and the explicit expressions of the right-hand sides of \eqref{theorem2-eq1}--\eqref{theorem2-eq2} are given in Appendix C.

\begin{remark}\label{rem3}
As pointed out in \cite{Yuan2018}, it is an open problem how to establish a reduced cumulative constraint violation bound for online convex optimization. Such a bound is first established by the distributed online algorithm with full-information feedback in \cite{Yi2024}. In bandit convex optimization, such a bound is still missing as the method that directly using the clipped constraint functions to establish network cumulative constraint violation bound as used in \cite{Yi2023} makes Slater’s condition ineffective. Theorem~2 establishes the reduced cumulative constraint violation bound \eqref{theorem2-eq2}, which is the same as the result established in \cite{Yi2024}, and thus fills the gap.
\end{remark}
\begin{remark}\label{rem4}
Slater’s condition plays an important role for establishing the reduced network cumulative constraint violation bound \eqref{theorem2-eq2}. We will provide an elucidation of why network cumulative constraint violation bounds can be reduced under Slater's condition and why directly using clipped constraint functions as done in \cite{Yi2023} makes Slater's condition ineffective. For the first question, it is worth noting that \eqref{lemma5-eqe} in Lemma~5 is very important to establish network cumulative constraint violation bounds. Without Slater's condition, from \eqref{lemma5-2-eq6}, we have $\frac{{q_{i,t + 1}^T{g_{i,t}}( {{x_{i,t}}} )}}{{{\gamma _t}}} \ge 0$, and then we have $\sum\nolimits_{i = 1}^n {\sum\nolimits_{t = 1}^T {\frac{{ {\mathbf{E}_{{\mathfrak{U}_t}}}[{{\| {\varepsilon _{i,t}^z} \|}^2}]}}{{4{\gamma _0}}}} }  \le {{\tilde h}_T}( {\hat y} )$ since \eqref{lemma5-eqb} and \eqref{lemma6-2-eq4} (i.e., $\sum\nolimits_{i = 1}^n {\sum\nolimits_{t = 1}^T {\frac{{q_{i,t + 1}^T{g_{i,t}}( y )}}{{{\gamma _t}}}} }  \le 0$) hold. Therefore, we have the result in \eqref{lemma6-2-eq5}. Based on the result, we can get \eqref{lemma6-eqb} in Lemma~6, and thus establish the network cumulative constraint violation bound $\mathcal{O}\big( {{{( {T\mathbf{E}[{{\tilde h}_T}( \hat{y} )} )]}^{1/2}}} \big)$ for the general cases. With Slater's condition, we have \eqref{theorem2-2-eq1} (i.e., $\sum\nolimits_{i = 1}^n {\sum\nolimits_{t = 1}^T {\frac{{q_{i,t + 1}^T{g_{i,t}}( x_s )}}{{{\gamma _t}}}} }  \le  - {\varsigma _s}\sum\nolimits_{i = 1}^n {\sum\nolimits_{t = 1}^T {\| {{{[ {{g_{i,t}}({x_{i,t}})} ]}_ + }} \|} }$). Note that the result in \eqref{theorem2-2-eq1} is tighter than that in \eqref{lemma6-2-eq4}. By using ${\frac{{ {\mathbf{E}_{{\mathfrak{U}_t}}}[{{\| {\varepsilon _{i,t}^z} \|}^2}]}}{{4{\gamma _0}}}} \ge 0$, we have a new lower bound ${\varsigma _s}\sum\nolimits_{i = 1}^n {\sum\nolimits_{t = 1}^T {\| {{{[ {{g_{i,t}}({x_{i,t}})} ]}_ + }} \|} }$ for ${{\tilde h}_T}( {\hat y} )$ to replace the lower bound $\sum\nolimits_{i = 1}^n {\sum\nolimits_{t = 1}^T {\frac{{ {\mathbf{E}_{{\mathfrak{U}_t}}}[{{\| {\varepsilon _{i,t}^z} \|}^2}]}}{{4{\gamma _0}}}} }$. As a result, by using \eqref{lemma6-eqc} in Lemma~6, we establish the network cumulative constraint violation bound $\mathcal{O}\big( \mathbf{E}[{{{\tilde h}_T}( {\hat y} )}] \big)$. By appropriately designing the stepsize sequence $\{ {{\alpha _t}} \}$, we can guarantee $\mathcal{O}\big( \mathbf{E}[{{{\tilde h}_T}( {\hat y} )}] \big) = \mathbf{o}( T )$. Therefore, the intuition is readily seen that $\mathcal{O}\big( \mathbf{E}[{{{\tilde h}_T}( {\hat y} )}] \big)$ is in general smaller with respect to $T$ than $\mathcal{O}\big( {{{( {T\mathbf{E}[{{\tilde h}_T}( \hat{y} )} )]}^{1/2}}} \big)$. Based on the above elucidation, we know that network cumulative constraint violation bounds can be reduced under Slater's condition. For the second question, it should be pointed out that \eqref{theorem2-2-eq1} is a key result to reduce network cumulative constraint violation bounds as discussed in the elucidation of the first question. However, if we directly use clipped constraint functions, the term $\frac{1}{n}\sum\nolimits_{i = 1}^n {q_{i,t + 1}^T{g_{i,t}}(y)}$ in \eqref{lemma4-eq1} would be replaced by $\frac{1}{n}\sum\nolimits_{i = 1}^n {q_{i,t + 1}^T{{[ {{g_{i,t}}(y)} ]}_ + }}$, and then the result in \eqref{theorem2-2-eq1} would become $\sum\nolimits_{i = 1}^n {\sum\nolimits_{t = 1}^T {\frac{{q_{i,t + 1}^T{{[ {{g_{i,t}}({x_s})} ]}_ + }}}{{{\gamma _t}}}} }  \ge 0$. Note that the property of Slater's condition in \eqref{ass8-eq1} does not work to establish the lower bound for ${{\tilde h}_T}( {\hat y} )$, i.e., Slater’s condition is ineffective. Based on the above elucidation, we know that directly using clipped constraint functions makes Slater’s condition ineffective.
\end{remark}
\subsection{Performance Analysis for Strongly Convex Case}
In the subsection, we establish network regret and cumulative constraint violation bounds for Algorithm~1 in the following theorems when local loss functions are strongly convex.

\begin{assumption}
For any $i \in [ n ]$ and $t \in {\mathbb{N}_ + }$, $\{ {{f_{i,t}}( x )} \}$ are strongly convex with convexity parameter $\mu  > 0$ over $\mathbb{X}$, i.e., for all $x,y \in \mathbb{X}$,
\begin{flalign}
{f_{i,t}}( x ) \ge {f_{i,t}}( y ) + \langle {x - y,\partial {f_{i,t}}( y )} \rangle  + \frac{\mu }{2}\| {x - y} \|^2. \label{ass9-eq1}
\end{flalign}
\end{assumption}

When the convex parameter $\mu$ is unknown, we use the natural vanishing stepsize sequence as in \eqref{theorem1-eq1} of Theorem~1.
\begin{theorem}\label{thm3}
Suppose Assumptions 1--4 and 6 hold. For all $i \in [ n ]$, let $\{ {{x_{i,t}}} \}$ be the sequences generated by Algorithm~1 with \eqref{theorem1-eq1}. Then, for any $T \in {\mathbb{N}_ + }$,
\begin{flalign}
\mathbf{E}[{{\rm{Net}\mbox{-}\rm{Reg}}( T )}] &= \mathcal{O}( {T^{1 - c} } ), \label{Theorem3-eq1} \\
\mathbf{E}[{{\rm{Net} \mbox{-} \rm{CCV}}( T )}] &= \mathcal{O}( {T^{1 - c/2}} ). \label{Theorem3-eq2}
\end{flalign}
Moreover, if Assumptions~5 also holds, then
\begin{flalign}
\mathbf{E}[{{\rm{Net} \mbox{-} \rm{CCV}}( T )}] &= \mathcal{O}( {T^{1 - c}} ). \label{Theorem3-eq3}
\end{flalign}
\end{theorem}
The proof and the explicit expressions of the right-hand sides of \eqref{Theorem3-eq1}--\eqref{Theorem3-eq3} are given in Appendix D.

\begin{remark}\label{rem5}
Theorem~3 shows that a reduced network regret bound \eqref{Theorem3-eq1} is established compare to the bounds \eqref{theorem1-eq2} and \eqref{theorem2-eq1} established in Theorems~1 and~2, respectively. The bounds \eqref{Theorem3-eq1} and \eqref{Theorem3-eq2} are the same as those established by the distributed bandit online algorithm in \cite{Yi2023}. It is worth noting that Algorithm~1 establishes a reduced network cumulative constraint violation bound \eqref{Theorem3-eq3} under Slater’s condition than the bound \eqref{Theorem3-eq2} without Slater's condition, while the algorithm in \cite{Yi2023} does not achieve such a result.
\end{remark}

When the convex parameter $\mu$ is known, we appropriately design the stepsize sequence in the following theorem.
\begin{theorem}\label{thm4}
Suppose Assumptions 1--4 and 6 hold. For all $i \in [ n ]$, let $\{ {{x_{i,t}}} \}$ be the sequences generated by Algorithm~1 with
\begin{flalign}
{\alpha _t} = \frac{1}{{\mu t}}, {\gamma _t} = \frac{{{\gamma _0}}}{{{\alpha _t}}}, {\xi _t} = {\alpha _t}, {\delta _t} = r( \mathbb{X} ){\alpha _t}, \label{theorem4-eq1}
\end{flalign}
where ${\gamma _0} \in \big( {0,1 /( {4(p^2+1)G_2^2} )} \big]$ is a constant. Then, for any $T \in {\mathbb{N}_ + }$,
\begin{flalign}
\mathbf{E}[{{\rm{Net}\mbox{-}\rm{Reg}}( T )}] &= \mathcal{O}( \log ( T ) ), \label{theorem4-eq2} \\
\mathbf{E}[{{\rm{Net} \mbox{-} \rm{CCV}}( T )}] &= \mathcal{O}( \sqrt {\log ( T )T} ). \label{theorem4-eq3}
\end{flalign}
Moreover, if Assumptions~5 also holds, then
\begin{flalign}
\mathbf{E}[{{\rm{Net} \mbox{-} \rm{CCV}}( T )}] &= \mathcal{O}( \log ( T ) ). \label{theorem4-eq4}
\end{flalign}
\end{theorem}
The proof and the explicit expressions of the right-hand sides of \eqref{theorem4-eq2}--\eqref{theorem4-eq4} are given in Appendix E.
\begin{remark}\label{rem6}
Theorem~4 shows that the reduced network regret bound \eqref{theorem4-eq2} and network cumulative constraint violation bounds \eqref{theorem4-eq3} and \eqref{theorem4-eq4} compare to the bounds \eqref{Theorem3-eq1}--\eqref{Theorem3-eq3}, respectively. It should be pointed out that the bound \eqref{theorem4-eq4} is achieved for the first time in the literature. In addition, the bounds \eqref{theorem4-eq2} and \eqref{theorem4-eq3} improve the results established by the distributed online algorithm with one-point bandit feedback in \cite{Yuan2022}, even though in \cite{Yuan2022} local constraint functions are static and known in advance.
\end{remark}

\section{NUMERICAL EXAMPLE}
To evaluate the performance of Algorithm~1, we consider a distributed online linear regression problem with time-varying linear inequality constraints over a network of $n$ agents.
At iteration~$t$, the local loss and constraint functions are  ${f_{i,t}}( x ) = \frac{1}{2}{( {{A_{i,t}}x - {\vartheta _{i,t}}} )^2}$ and ${g_{i,t}}( x ) = {B_{i,t}}x - {b_{i,t}}$, respectively, where each component of ${A_{i,t}} \in {\mathbb{R}^{{q_i} \times p}}$ is randomly generated from the uniform distribution in the interval $[ { - 1,1} ]$, ${\vartheta _{i,t}} = {A_{i,t}}{\mathbf{1}_p} + {\zeta _{i,t}}$ with ${\vartheta _{i,t}} \in {\mathbb{R}^{{q_i}}}$ and ${\zeta _{i,t}}$ being a standard normal random vector, each component of ${B_{i,t}} \in {\mathbb{R}^{{m_i} \times p}}$ is randomly generated from the uniform distribution in the interval $[ {0,2} ]$, each component of ${b_{i,t}} \in {\mathbb{R}^{{m_i}}}$ is randomly generated from the uniform distribution in the interval $[ {b,b+1} ]$ with $b > 0$. Note that $b > 0$ guarantees Slater’s condition holds.
We use an time-varying undirected graph to model the communication topology.
Specifically, at each iteration~$t$, the graph is first randomly generated where the probability of any two agents being connected is $\rho $. Then, to make sure that Assumption~4 is satisfied, we add edges~$( {i,i + 1} )$ for $i \in [ 24 ]$ when $t \in \{ {4c + 1} \}$, edges~$( {i,i + 1} )$ for $i \in [ 25, 49 ]$ when $t \in \{ {4c + 2} \}$, edges~$( {i,i + 1} )$ for $i \in [ 50, 74 ]$ when $t \in \{ {4c + 3} \}$, edges~$( {i,i + 1} )$ for $i \in [ 75, 99 ]$ when $t \in \{ {4c + 4} \}$ for $c = \{0, 1, \cdot  \cdot  \cdot \}$. Moreover, let
${[ {{W_t}} ]_{ij}} = \frac{1}{n}$ if $( {j,i} ) \in {\mathcal{E}_t}$ and ${[ {{W_t}} ]_{ii}} = 1 - \sum\nolimits_{j = 1}^n {{{[ {{W_t}} ]}_{ij}}} $.

In this paper, without Slater’s condition, we show that Algorithm~1 establishes the same network regret and cumulative constraint violation bounds as those in \cite{Yi2023}. More importantly, with Slater’s condition, we show that Algorithm~1 establishes the reduced network cumulative constraint violation bounds, which is significant results not found in existing literature. To verify our theoretical results, we compare Algorithm~1 with the distributed online algorithm with two-point bandit feedback in \cite{Yi2023} that uses the cumulative constraint violation metric but does not consider Slater’s condition, and the distributed online algorithm with full-information feedback in \cite{Yi2024} that uses the cumulative constraint violation metric and consider Slater’s condition.
We set $n = 100$, ${q_i} = 4$, $p = 10$, ${m_i} = 2$, $\mathbb{X} = {[ { - 5,5} ]^p}$, $b = 0.01$, and $\rho  = 0.1$. The inputs
of the algorithms are listed in TABLE~II.
\begin{table*}
\centering
\caption{Input of algorithms.}
\begin{tabular}{c|c|c}
\Xcline{1-3}{1pt}
\multicolumn{2}{c|}{\multirow{2}{*}{Algorithms}} & {\multirow{2}{*}{Inputs}}\\
\multicolumn{2}{c|}{} & \\
\cline{1-3}
\multicolumn{2}{c|}{\multirow{2}{*}{Algorithm~1 in this paper}} & {\multirow{2}{*}{${\alpha _t} = 1/t,{\gamma _t} = 0.15/{\alpha _t},{\xi _t} = 1/t,{\delta _t} = 0.01/t$}}\\
\multicolumn{2}{c|}{} & \\
\cline{1-3}
\multicolumn{2}{c|}{\multirow{2}{*}{Algorithm~2 in \cite{Yi2023}}} & {\multirow{2}{*}{${\alpha _t} = 1/{t},{\beta _t} = 1/{t^{0.5}},{\gamma _t} = 1/{t^{0.5}},{\xi _t} = 1/( {t + 1} ),{\delta _t} = 0.01/( {t + 1} )$}}\\
\multicolumn{2}{c|}{} & \\
\cline{1-3}
\multicolumn{2}{c|}{\multirow{2}{*}{Algorithm~1 in \cite{Yi2024}}} & {\multirow{2}{*}{${\alpha _t} = 2/t,{\gamma _t} = 0.15/{\alpha _t},\psi ( x ) = {\| x \|^2}$}}\\
\multicolumn{2}{c|}{} & \\
\Xcline{1-3}{1pt}
\end{tabular}
\end{table*}
\begin{figure}[!ht]
 \centering
  \includegraphics[width=10cm]{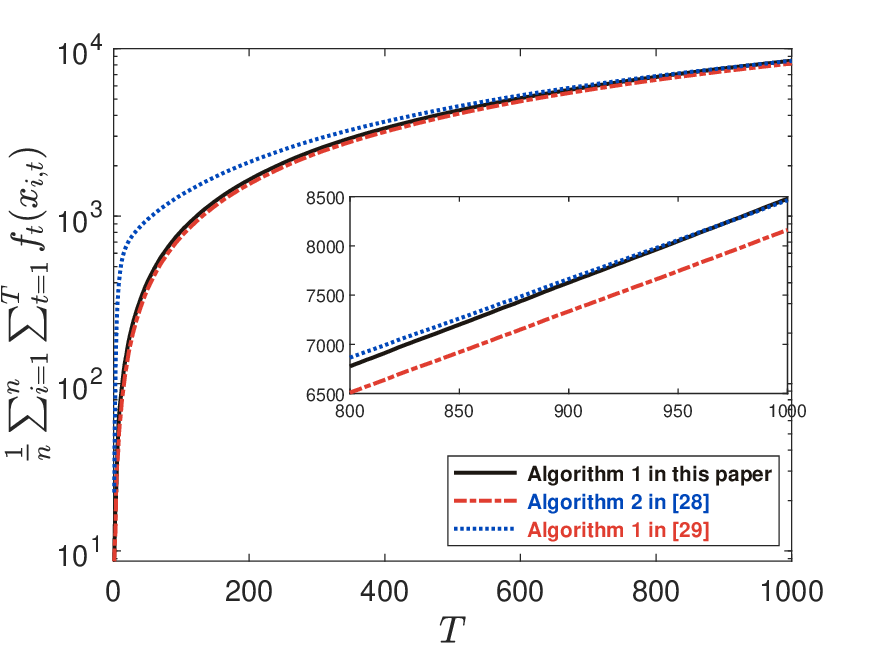}
  \caption{Evolutions of $\frac{1}{n}\sum\nolimits_{i = 1}^n {\sum\nolimits_{t = 1}^T {{f_t}( {{x_{i,t}}} )} }$.}
\end{figure}

\begin{figure}[!ht]
 \centering
  \includegraphics[width=10cm]{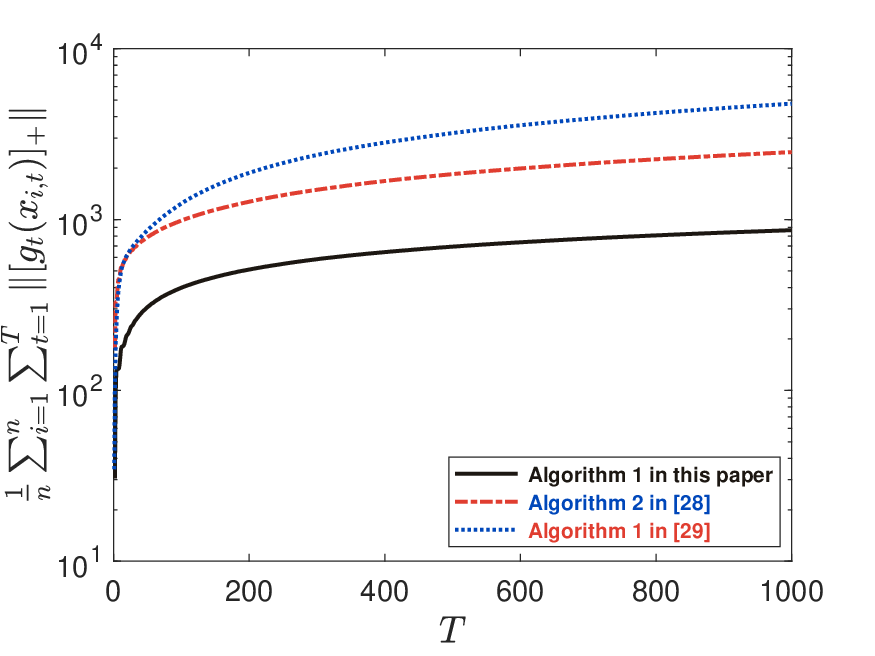}
  \caption{Evolutions of $\frac{1}{n}\sum\nolimits_{i = 1}^n {\sum\nolimits_{t = 1}^T {\| {{{[ {{g_t}( {{x_{i,t}}} )} ]}_ + }} \|} } $.}
\end{figure}

Figs.~1 and 2 illustrate the evolutions of the cumulative loss $\frac{1}{n}\sum\nolimits_{i = 1}^n {\sum\nolimits_{t = 1}^T {{f_t}( {{x_{i,t}}} )} }$ and the cumulative constraint violation $\frac{1}{n}\sum\nolimits_{i = 1}^n {\sum\nolimits_{t = 1}^T {\| {{{[ {{g_t}( {{x_{i,t}}} )} ]}_ + }} \|} }$, respectively.
Fig. 1 demonstrates that our Algorithm~1 has almost the same accumulated loss as that of Algorithm~2 in \cite{Yi2023}, but has slightly larger accumulated loss than that of Algorithm~1 in \cite{Yi2024}. That is reasonable since Algorithm~1 in \cite{Yi2024} directly uses the subgradients of local loss and constraint functions while our Algorithm~1 uses two-point stochastic estimators to approximate these subgradients. 
Fig. 2 demonstrates that our Algorithm~1 has significantly smaller cumulative constraint violation than that of Algorithm~2 in \cite{Yi2023}, which are consistent with the theoretical results in Theorem~4. The key reason is that Slater’s condition remains effective in our Algorithm~1 but becomes ineffective in Algorithm~2 in \cite{Yi2023}. 

\section{CONCLUSIONS}
This paper studied the distributed bandit convex optimization problem with time-varying inequality constraints.
We proposed a new distributed bandit online primal--dual algorithm, and established network regret and cumulative constraint violation bounds for convex and strongly convex cases, respectively.
Without Slater’s condition, the bounds were the same as the state-of-the-art those in the literature. With Slater’s condition, the network cumulative constraint violation bounds were reduced.
In the future, we will explore the scenario where communication resources are limited, and investigate distributed bandit online algorithms with compressed communication.

\appendix

\hspace{-3mm}\emph{A. Useful Lemmas}

Some preliminary results are given in this subsection.
We first provide some results on the projection in the following lemma.
\begin{lemma}\label{lem1}
Let $\mathbb{K}$ be a nonempty closed convex subset of ${\mathbb{R}^p}$ and let $a$ and $b$ be two vectors in ${\mathbb{R}^p}$.  
If ${x_c} = {\mathcal{P}_\mathbb{K}}( {b - a} )$, then
\begin{flalign}
2\langle {{x_c} - y,a} \rangle  \le {\| {y - b} \|^2} - {\| {y - {x_c}} \|^2} - {\| {{x_c} - b} \|^2},\forall y \in \mathbb{K}. \label{lemma1-eq1a}
\end{flalign}
In addition, let $\Phi ( y ) = {\| {b - y} \|^2} + 2\langle {a,y} \rangle$, then we know $\Phi$ is a strongly convex function with convexity parameter $\sigma = 2$ and ${x_c} = \mathop {\arg \min }\limits_{y \in \mathbb{K}} \Phi ( y )$. We can obtain
\begin{flalign}
\| {{x_c} - b} \| \le \| a \|. \label{lemma1-eq1b}
\end{flalign}
\end{lemma}
\begin{proof}
First, we know \eqref{lemma1-eq1a} holds from Lemma~3 in \cite{Yi2021b}.
Then, 
Since $\Phi$ is a strongly convex function with convexity parameter $\sigma = 2$, we have
\begin{flalign}
\nonumber
\Phi ( b ) \ge \Phi ( {{x_c}} ) + {\big( {\nabla \Phi ( {{x_c}} )} \big)^T}( {b - {x_c}} ) + \frac{\sigma }{2}\| {b - {x_c}} \|.
\end{flalign}
From the optimality condition, we have
\begin{flalign}
\nonumber
{\big( {\nabla \Phi ( {{x_c}} )} \big)^T}( {b - {x_c}} ) \ge 0.
\end{flalign}
Thus, we have
\begin{flalign}
\nonumber
\Phi ( b ) \ge \Phi ( {{x_c}} ) + \frac{\sigma }{2}\| {b - {x_c}} \|.
\end{flalign}
From $\Phi ( y ) = {\| {b - y} \|^2} + 2\langle {a,y} \rangle$, we have
\begin{flalign}
\nonumber
2\langle {a,b} \rangle  \ge {\| {b - {x_c}} \|^2} + 2\langle {a,{x_c}} \rangle  + \frac{\sigma }{2}{\| {b - {x_c}} \|^2}.
\end{flalign}
Thus, we have
\begin{flalign}
\nonumber
2\langle {a,b - {x_c}} \rangle  \ge {\| {b - {x_c}} \|^2} + \frac{\sigma }{2}{\| {b - {x_c}} \|^2}.
\end{flalign}
From the Cauchy-Schwarz inequality, we have
\begin{flalign}
\nonumber
2\langle {a,b - {x_c}} \rangle  \le 2\| a \|\| {b - {x_c}} \|.
\end{flalign}
Thus, we have
\begin{flalign}
\nonumber
{\| {b - {x_c}} \|^2} + \frac{\sigma }{2}{\| {b - {x_c}} \|^2} \le 2\| a \|\| {b - {x_c}} \|.
\end{flalign}
Due to $\sigma = 2$, we have
\begin{flalign}
\nonumber
\| {b - {x_c}} \| \le \| a \|.
\end{flalign}
Therefore, we know that \eqref{lemma1-eq1b} holds.
\end{proof}

We then present some properties of the subgradient estimators $\partial {{\hat f}_{i,t}}$ and $\partial {{{{\hat g}_{i,t}}} }$ in the following lemma.
\begin{lemma}\label{lem1}
(Lemma~8 in \cite{Yi2023})
If Assumption~3 holds. Then, ${{\hat f}_{i,t}}( x )$ and ${{{{\hat g}_{i,t}}( x )} }$ are convex on $( {1 - {\xi _t}} )\mathbb{X}$. If ${{f}_{i,t}}( x )$ and ${{{{g}_{i,t}}( x )} }$ are strongly convex with convexity parameter  $\mu  > 0$ over $\mathbb{X}$, Then, ${{\hat f}_{i,t}}( x )$ and ${{{{\hat g}_{i,t}}( x )} }$ are strongly convex with convexity parameter  $\mu  > 0$ over $( {1 - {\xi _t}} )\mathbb{X}$. Moreover, for any $i \in [ n ]$, $t \in {\mathbb{N}_ + }$, $x \in ( {1 - {\xi _t}} )\mathbb{X}$, $q \in \mathbb{R}_ + ^{{m_i}}$,
\begin{subequations}
\begin{flalign}
&\partial {{\hat f}_{i,t}}( x ) = {\mathbf{E}_{{\mathfrak{U}_t}}}[ {\hat \partial {f_{i,t}}( x )} ], \label{lemma2-eq1a}\\
&{f_{i,t}}( x ) \le {{\hat f}_{i,t}}( x ) \le {f_{i,t}}( x ) + {G_1}{\delta _t}, \label{lemma2-eq1b}\\
&\| {\hat \partial {f_{i,t}}( x )} \| \le p{G_1}, \label{lemma2-eq1c}\\
&\partial {{{{\hat g}_{i,t}}( x )} } = {\mathbf{E}_{{\mathfrak{U}_t}}}\big[ {\hat \partial {{{g_{i,t}}( x )} }} \big], \label{lemma2-eq1d} \\
&{q^T}{{{g_{i,t}}( x )} } \le {q^T}{{{{\hat g}_{i,t}}( x )} }
\le {q^T}{{{g_{i,t}}( x )} } + {G_2}{\delta _t}\| q \|, \label{lemma2-eq1e}\\
&\| {\hat \partial {{{g_{i,t}}( x )} }} \| \le p{G_2}, \label{lemma2-eq1f}
\end{flalign}
\end{subequations}
where ${{\hat f}_{i,t}}( x ) = {\mathbf{E}_{v \in {\mathbb{B}^p}}}[ {{f_{i,t}}( {x + {\delta _t}v} )} ]$ and ${{{{\hat g}_{i,t}}( x )} } = {\mathbf{E}_{v \in {\mathbb{B}^p}}}\big[ {{{{g_{i,t}}( {x + {\delta _t}v} )}}} \big]$ with $v$ being chosen uniformly at random, and ${\mathfrak{U}_t}$ is the $\sigma$-algebra induced by the independent and identically distributed variables ${u_{1,t}}, \cdot  \cdot  \cdot ,{u_{n,t}}$.
\end{lemma}

We next quantify the disagreement among the local temporary primal variables $\{ {{z_{i,t}}} \}$.
\begin{lemma}\label{lem2}
(Lemma 4 in \cite{Yi2023})
If Assumption 4 holds. For all $i \in [ n ]$ and $t \in {\mathbb{N}_ + }$, ${ z_{i,t}}$ generated by Algorithm~1 satisfy
\begin{flalign}
\| {{z_{i,t}} - {{\bar z}_t}} \| &\le \tau {\lambda ^{t - 2}}\sum\limits_{j = 1}^n {\| {{z_{j,1}}} \|}  + \frac{1}{n}\sum\limits_{j = 1}^n {\| {\varepsilon _{j,t - 1}^z} \|}  + \| {\varepsilon _{i,t - 1}^z} \|
+ \tau \sum\limits_{s = 1}^{t - 2} {{\lambda ^{t - s - 2}}} \sum\limits_{j = 1}^n {\| {\varepsilon _{j,s}^z} \|}, \label{lemma3-eq1}
\end{flalign}
where ${{\bar z}_t} = \frac{1}{n}\sum\nolimits_{j = 1}^n {{z_{j,t}}} $, $\tau = {( {1 - \omega /4{n^2}} )^{ - 2}} > 1$, $\lambda  = {( {1 - \omega /4{n^2}} )^{1/B}} \in ( {0,1} )$, and $\varepsilon _{i,t - 1}^z = {z_{i,{t}}} - {x_{i,t-1}}$.
\end{lemma}

We finally analyze regret at one iteration.
\begin{lemma}\label{lem4}
Suppose Assumptions 1--3 hold. For all $i \in [ n ]$, let $\{ {{x_{i,t}}}\}$ be the sequences generated by Algorithm~1 and $\{ {{y_t}} \}$ be an arbitrary sequence in $\mathbb{X}$, then
\begin{flalign}
\nonumber
&\;\;\;\;\;\frac{1}{n}\sum\limits_{i = 1}^n {q_{i,t + 1}^T{g_{i,t}}( {{x_{i,t}}} )}  + \frac{1}{n}\sum\limits_{i = 1}^n {\big( {{f_{i,t}}( {{x_{i,t}}} ) - {f_{i,t}}( y )} \big)} \\
\nonumber
& \le \frac{1}{n}\sum\limits_{i = 1}^n {q_{i,t + 1}^T{g_{i,t}}( y )} + \frac{1}{n}\sum\limits_{i = 1}^n {{\Delta _{i,t}}( {\hat y} )} + \frac{{{{\tilde \Delta }_t}}}{n} \\
& + \frac{1}{n}\sum\limits_{i = 1}^n {{G_2}} \big( {R( \mathbb{X} ){\xi _t} + {\delta _t}} \big)\| {{q_{i,t + 1}}} \| + {G_1}R( \mathbb{X} ){\xi _t} + {G_1}{\delta _t} , \label{lemma4-eq1}
\end{flalign}
where
\begin{flalign}
\nonumber
&{\Delta _{i,t}}( {\hat y} ) = \frac{1}{{{2\alpha _t}}}{\mathbf{E}_{{\mathfrak{U}_t}}}\big[{{{\| {\hat y - {x_{i,t}}} \|}^2} - {{\| {\hat y - {x_{i,t + 1}}} \|}^2}} \big], \\
\nonumber
&{{\tilde \Delta }_t} = \sum\limits_{i = 1}^n {( {p{G_1} + p{G_2}\| {{q_{i,t + 1}}} \|} ){\mathbf{E}_{{\mathfrak{U}_t}}}[\| {\varepsilon _{i,t}^z} \|]}  - \sum\limits_{i = 1}^n {\frac{{{\mathbf{E}_{{\mathfrak{U}_t}}}[{{\| {\varepsilon _{i,t}^z} \|}^2}]}}{{2{\alpha _t}}}}.
\end{flalign}
\end{lemma}
\begin{proof}
From Assumption~3, we have
\begin{flalign}
{f_{i,t}}( y ) &\ge {f_{i,t}}( x ) + \langle {\partial {f_{i,t}}( x ),y - x} \rangle ,\forall x,y \in \mathbb{X}, \label{lemma4-eq2}\\
{g_{i,t}}( y ) &\ge {g_{i,t}}( x ) + \partial {g_{i,t}}( x )( {y - x} ),\forall x,y \in \mathbb{X}. \label{lemma4-eq4}
\end{flalign}

From \eqref{ass1-eq1}, \eqref{ass4-eq2a}, \eqref{lemma2-eq1b} and $\hat y = ( {1 - {\xi _t}} )y$ for any $t \in {\mathbb{N}_ + }$, we have
\begin{flalign}
\nonumber
{{\hat f}_{i,t}}( {\hat y} ) - {f_{i,t}}( y ) &= {f_{i,t}}( {\hat y} ) - {f_{i,t}}( y ) + {{\hat f}_{i,t}}( {\hat y} ) - {f_{i,t}}( {\hat y} ) \\
\nonumber
& \le {G_1}\| {\hat y - y} \| + {{\hat f}_{i,t}}( {\hat y} ) - {f_{i,t}}( {\hat y} ) \\
& \le {G_1}R( \mathbb{X} ){\xi _t} + {G_1}{\delta _t}. \label{lemma4-eq6}
\end{flalign}
We have
\begin{flalign}
\nonumber
q_{i,t + 1}^T{{\hat g}_{i,t}}( {\hat y} ) &\le q_{i,t + 1}^T{g_{i,t}}( {\hat y} ) + {G_2}{\delta _t}\| {{q_{i,t + 1}}} \| \\
\nonumber
& = q_{i,t + 1}^T\big( {{g_{i,t}}( y ) + {g_{i,t}}( {\hat y} ) - {g_{i,t}}( y )} \big) + {G_2}{\delta _t}\| {{q_{i,t + 1}}} \| \\
\nonumber
& \le q_{i,t + 1}^T{g_{i,t}}( y ) + {G_2}\| {\hat y - y} \|\| {{q_{i,t + 1}}} \| + {G_2}{\delta _t}\| {{q_{i,t + 1}}} \| \\
& \le q_{i,t + 1}^T{g_{i,t}}( y ) + {G_2}\big( {R( \mathbb{X} ){\xi _t} + {\delta _t}} \big)\| {{q_{i,t + 1}}} \|, \label{lemma4-eq7}
\end{flalign}
where the first inequality holds due to \eqref{lemma2-eq1e}; the second inequality holds due to \eqref{ass4-eq1c} and \eqref{lemma4-eq4}; and the last inequality holds due to \eqref{ass1-eq1} and $\hat y = \left( {1 - {\xi _t}} \right)y$.

We have
\begin{flalign}
\nonumber
{{\hat f}_{i,t}}( {{x_{i,t}}} ) - {{\hat f}_{i,t}}( {\hat y} )
& \le \langle {\partial {{\hat f}_{i,t}}( {{x_{i,t}}} ),{x_{i,t}} - \hat y} \rangle \\
\nonumber
&  = \langle {{\mathbf{E}_{{\mathfrak{U}_t}}}[ {\hat \partial {f_{i,t}}( {{x_{i,t}}} )} ],{x_{i,t}} - \hat y} \rangle \\
\nonumber
&  = {\mathbf{E}_{{\mathfrak{U}_t}}}[ {\langle {\hat \partial {f_{i,t}}( {{x_{i,t}}} ),{x_{i,t}} - \hat y} \rangle } ] \\
\nonumber
&  = {\mathbf{E}_{{\mathfrak{U}_t}}}[ {\langle {\hat \partial {f_{i,t}}( {{x_{i,t}}} ),{x_{i,t}} - z_{i,t + 1}} \rangle  + \langle {\hat \partial {f_{i,t}}( {{x_{i,t}}} ),{z_{i,t + 1}} - \hat y} \rangle } ] \\
& \le {\mathbf{E}_{{\mathfrak{U}_t}}}[ {p{G_1}\| {\varepsilon _{i,t}^z} \| + \langle {\hat \partial {f_{i,t}}( {{x_{i,t}}} ),{z_{i,t + 1}} - \hat y} \rangle } ], \label{lemma4-eq8}
\end{flalign}
where the first inequality holds since ${x_{i,t}}$, $\hat y \in ( {1 - {\xi _t}} )\mathbb{X}$ and ${f_{i,t}}( x )$ is convex on $( {1 - {\xi _t}} )\mathbb{X}$; the first equality holds due to \eqref{lemma2-eq1a}; the second equality holds since ${x_{i,t}}$ and $\hat y$ are independent of ${\mathfrak{U}_t}$; and the last inequality holds due to \eqref{lemma2-eq1c}.

From \eqref{Algorithm1-eq3}, we have
\begin{flalign}
\nonumber
\langle {\hat \partial {f_{i,t}}( {{x_{i,t}}} ),{z_{i,t + 1}} - \hat y} \rangle
&  = \langle {{{\hat \omega }_{i,t + 1}},{z_{i,t + 1}} - \hat y} \rangle  + \langle {\hat \partial {g_{i,t}}( {{x_{i,t}}} ){q_{i,t + 1}},\hat y - {z_{i,t + 1}}} \rangle \\
\nonumber
&  = \langle {{{\hat \omega }_{i,t + 1}},{z_{i,t + 1}} - \hat y} \rangle  + \langle {\hat \partial {g_{i,t}}( {{x_{i,t}}} ){q_{i,t + 1}},\hat y - {x_{i,t}}} \rangle \\
&\;\;\;\;  + \langle {\hat \partial {g_{i,t}}( {{x_{i,t}}} ){q_{i,t + 1}},{x_{i,t}} - {z_{i,t + 1}}} \rangle. \label{lemma4-eq9}
\end{flalign}

We have
\begin{flalign}
\nonumber
{\mathbf{E}_{{\mathfrak{U}_t}}}[ {\langle {\hat \partial {g_{i,t}}( {{x_{i,t}}} ){q_{i,t + 1}},\hat y - {x_{i,t}}} \rangle } ]
&  = \langle {{\mathbf{E}_{{\mathfrak{U}_t}}}[ {\hat \partial {g_{i,t}}( {{x_{i,t}}} )} ]{q_{i,t + 1}},\hat y - {x_{i,t}}} \rangle  \\
\nonumber
&  = \langle {\partial {{\hat g}_{i,t}}( {{x_{i,t}}} ){q_{i,t + 1}},\hat y - {x_{i,t}}} \rangle \\
\nonumber
& \le q_{i,t + 1}^T{{\hat g}_{i,t}}( {\hat y} ) - q_{i,t + 1}^T{{\hat g}_{i,t}}( {{x_{i,t}}} ) \\
& \le q_{i,t + 1}^T{g_{i,t}}( y ) - q_{i,t + 1}^T{g_{i,t}}( {{x_{i,t}}} ) + {G_2}\big( {R( \mathbb{X} ){\xi _t} + {\delta _t}} \big)\| {{q_{i,t + 1}}} \|, \label{lemma4-eq10}
\end{flalign}
where the first equality holds since ${x_{i,t}}$, $q_{i,t + 1}$ and $\hat y$ are independent of ${\mathfrak{U}_t}$; the last equality holds due to \eqref{lemma2-eq1d}; the first inequality holds since ${q_{i,t}} \ge {\mathbf{0}_{{m_i}}}$, ${x_{i,t}},\hat y \in ( {1 - {\xi _t}} )\mathbb{X}$ and ${{\hat g}_{i,t}}( x )$ is convex on $( {1 - {\xi _t}} )\mathbb{X}$ as shown in Lemma~2; and the last inequality holds due to \eqref{lemma2-eq1e} and \eqref{lemma4-eq7}.

From the Cauchy-Schwarz inequality and \eqref{lemma2-eq1f}, we have
\begin{flalign}
\big\langle {{ {\hat{\partial} {g_{i,t}}( {{x_{i,t}}} )} }{q_{i,t + 1}},{x_{i,t}} - {z_{i,t + 1}}} \big\rangle  \le p{G_2}\| {{q_{i,t + 1}}} \|\| {\varepsilon _{i,t}^z} \|. \label{lemma4-eq11}
\end{flalign}

Applying \eqref{lemma1-eq1a} to update \eqref{Algorithm1-eq4}, we have
\begin{flalign}
\nonumber
&\;\;\;\;\;\langle {{\hat{\omega} _{i,t + 1}}, {z_{i,t + 1}} - \hat y} \rangle \\
\nonumber
& \le \frac{1}{{2{\alpha _t}}}\big( {{{\| {\hat y - {x_{i,t}}} \|}^2} - {{\| {\hat y - {z_{i,t + 1}}} \|}^2} - \| {\varepsilon _{i,t}^z} \|^2} \big) \\
\nonumber
& = \frac{1}{{2{\alpha _t}}}\big( {{{\| {\hat y - {x_{i,t}}} \|}^2} - {{\| {\hat y - {x_{i,t + 1}}} \|}^2} + {{\| {\hat y - {x_{i,t + 1}}} \|}^2} - {{\| {\hat y - {z_{i,t + 1}}} \|}^2} - \| {\varepsilon _{i,t}^z} \|^2} \big) \\
\nonumber
& = {\Delta _{i,t}}( {\hat y} ) + \frac{1}{{2{\alpha _t}}}\Big( {{\Big\| {\hat y - \sum\limits_{j = 1}^n {{{[ {{W_{t + 1}}} ]}_{ij}}{z_{j,t + 1}}} } \Big\|^2} - {{\| {\hat y - {z_{i,t + 1}}} \|}^2} - \| {\varepsilon _{i,t}^z} \|^2} \Big) \\
& \le {\Delta _{i,t}}( {\hat y} ) + \frac{1}{{2{\alpha _t}}}\Big( {\sum\limits_{j = 1}^n {{{[ {{W_{t + 1}}} ]}_{ij}}{{\| {\hat y - {z_{j,t + 1}}} \|}^2}}  - {{\| {\hat y - {z_{i,t + 1}}} \|}^2} - \| {\varepsilon _{i,t}^z} \|^2} \Big), \label{lemma4-eq12}
\end{flalign}
where the second equality holds due to \eqref{Algorithm1-eq1}; and the last inequality holds since ${W_{t + 1}}$ is doubly stochastic and ${\|  \cdot  \|^2}$ are convex.

Summing \eqref{lemma2-eq1b}, \eqref{lemma4-eq6}, \eqref{lemma4-eq8}--\eqref{lemma4-eq12} over $i \in [ n ]$, dividing by $n$, using $\sum\nolimits_{i = 1}^n {{{[ {{W_t}} ]}_{ij}} = 1} ,\forall t \in {\mathbb{N}_ + }$, and rearranging terms yields \eqref{lemma4-eq1}.
\end{proof}

Finally, we bound local regret and (squared) cumulative constraint violation, the accumulated (squared) consensus error, and the changes caused by composite mirror descent in the following.
\begin{lemma}\label{lem5}
Suppose Assumptions 1--2 and 4--6 hold. For all $i \in [ n ]$, let $\{ {{x_{i,t}}}\}$ be the sequences generated by Algorithm~1 with ${\gamma _t} = {\gamma _0}/{\alpha _t}$, where ${\gamma _0} \in \big( {0,1 /( {4(p^2+1)G_2^2} )} \big]$ is a constant. Then, for any $T \in {\mathbb{N}_ + }$,
\begin{flalign}
&\frac{1}{n}\sum\limits_{i = 1}^n {\sum\limits_{t = 1}^T {\big( {{f_{i,t}}( {{x_{i,t}}} ) - {f_{i,t}}( y ) + \frac{{ {\mathbf{E}_{{\mathfrak{U}_t}}}[{{\| {\varepsilon _{i,t}^z} \|}^2}]}}{{4{\alpha _t}}}} \big)} } \le {m_T} + \frac{1}{n}\sum\limits_{i = 1}^n {\sum\limits_{t = 1}^T {{\Delta _{i,t}}( {\hat y} )} }, \label{lemma5-eqa}\\
&\sum\limits_{i = 1}^n {\sum\limits_{t = 1}^T {\frac{1}{2}} } \Big( {\frac{{q_{i,t + 1}^T{g_{i,t}}( {{x_{i,t}}} )}}{{{\gamma _t}}} + \frac{{\ {\mathbf{E}_{{\mathfrak{U}_t}}}[{{\| {\varepsilon _{i,t}^z} \|}^2}]}}{{2{\gamma _0}}}} \Big) \le {h_T}( y ) + {{\tilde h}_T}( {\hat y} ), \label{lemma5-eqb}\\
&\frac{1}{n}\sum\limits_{t = 1}^T {\sum\limits_{i = 1}^n {\sum\limits_{j = 1}^n {\| {{x_{i,t}} - {x_{j,t}}} \|} } } \le n{\tilde{\varepsilon} _1} + {\tilde \varepsilon _2}\sum\limits_{t = 1}^T {\sum\limits_{i = 1}^n {\| {\varepsilon _{i,t}^z} \|} }, \label{lemma5-eqc}\\
&\frac{1}{n}\sum\limits_{t = 1}^T {\sum\limits_{i = 1}^n {\sum\limits_{j = 1}^n {{{\| {{x_{i,t}} - {x_{j,t}}} \|}^2}} } } \le {{\tilde \varepsilon }_3} + {{\tilde \varepsilon }_4}\sum\limits_{t = 1}^T {\sum\limits_{i = 1}^n {{{\| {\varepsilon _{i,t}^z} \|}^2}} }, \label{lemma5-eqd}\\
&{\mathbf{E}_{{\mathfrak{U}_t}}}[\| {\varepsilon _{i,t}^z} \|] \le {p{G_1}{\alpha _t} + p{G_2}{\gamma _0}\| {{{[ {{g_{i,t}}( {{x_{i,t}}} )} ]}_ + }} \|}. \label{lemma5-eqe}
\end{flalign}
where
\begin{flalign}
\nonumber
&{m_T} = \sum\limits_{t = 1}^T {{2{p^2{G_1^2}}{\alpha _t}}}  + \sum\limits_{t = 1}^T {\frac{{ R{{(\mathbb{X})}^2}\xi _t^2}}{{4{\alpha _t}}}}  + \sum\limits_{t = 1}^T {\frac{{ \delta _t^2}}{{4{\alpha _t}}}}  + \sum\limits_{t = 1}^T {{G_1}R(\mathbb{X}){\xi _t}}  + \sum\limits_{t = 1}^T {{G_1}{\delta _t}}, \\
\nonumber
&{h_T}( y ) = \sum\limits_{i = 1}^n {\sum\limits_{t = 1}^T {\frac{{q_{i,t + 1}^T{g_{i,t}}( y )}}{{{\gamma _t}}}} }, \\
\nonumber
&{{\tilde h}_T}( {\hat y} ) = \sum\limits_{i = 1}^n \frac{{\mathbf{E}_{{\mathfrak{U}_t}}}[{{{\| {\hat y - {x_{i,1}}} \|}^2}}]}{{{2\gamma _0}}}  + \sum\limits_{t = 1}^T {\frac{{2nF}}{{{\gamma _t}}}} + \sum\limits_{t = 1}^T {\frac{{2n{p^2{G_1^2}}{\gamma _0}}}{{\gamma _t^2}}} + \sum\limits_{t = 1}^T {\frac{{nR{{( \mathbb{X} )}^2}\xi _t^2}}{{4{\gamma _0}}}}  + \sum\limits_{t = 1}^T {\frac{{n\delta _t^2}}{{4{\gamma _0}}}} \\
\nonumber
&\;\;\;\;\;\;\;\;\;\; + \sum\limits_{t = 1}^T {\frac{{n{G_1}R( \mathbb{X} ){\xi _t} + n{G_1}{\delta _t}}}{{{\gamma _t}}}}, {\tilde{\varepsilon} _1} = \frac{{2\tau }}{{\lambda ( {1 - \lambda } )}}\sum\limits_{i = 1}^n {\| {{z_{i,1}}} \|}, \\
\nonumber
&{{\tilde \varepsilon }_2} = \frac{{2( {n\tau  + 2 - 2\lambda } )}}{{1 - \lambda }}, {{\tilde \varepsilon }_3} = \frac{{16n{\tau ^2}}}{{{\lambda ^2}( {1 - {\lambda ^2}} )}}{\Big( {\sum\limits_{i = 1}^n {\| {{z_{i,1}}} \|} } \Big)^2}, {{\tilde \varepsilon }_4} = \frac{{16{n^2}{\tau ^2}}}{{{{( {1 - \lambda } )}^2}}} + 32,
\end{flalign}
\end{lemma}
\begin{proof}
We will show that \eqref{lemma5-eqa}--\eqref{lemma5-eqe} hold in the following, respectively.

\noindent (i) Noting that ${g_{i,t}}( y ) \le {\mathbf{0}_{{m_i}}}$, $\forall i \in [ n ]$, $\forall t \in {\mathbb{N}_ + }$ when $y \in {\mathcal{X}_T}$, summing \eqref{lemma4-eq1} over $t \in [ T ]$ gives
\begin{flalign}
\nonumber
&\;\;\;\;\;\frac{1}{n}\sum\limits_{i = 1}^n {\sum\limits_{t = 1}^T {\big( {{f_{i,t}}( {{x_{i,t}}} ) - {f_{i,t}}( {{y}} )} \big)} } \\
\nonumber
& \le \frac{1}{n}\sum\limits_{i = 1}^n \sum\limits_{t = 1}^T \Big(-q_{i,t + 1}^T{g_{i,t}}( {{x_{i,t}}} ) + \frac{1}{n}{{\tilde \Delta }_t} + {\Delta _{i,t}}( {\hat y} ) \\
&\;\;\; + {G_2}\big( {R( \mathbb{X} ){\xi _t} + {\delta _t}} \big)\| {{q_{i,t + 1}}} \| + {G_1}R( \mathbb{X} ){\xi _t} + {G_1}{\delta _t} \Big). \label{lemma5-eq1}
\end{flalign}

We have
\begin{flalign}
\nonumber
&\;\;\;\;\; \sum\limits_{t = 1}^T {\sum\limits_{i = 1}^n {( {p{G_1} + p{G_2}\| {{q_{i,t + 1}}} \|} ){\mathbf{E}_{{\mathfrak{U}_t}}}[\| {\varepsilon _{i,t}^z} \|]} } \\
& \le \sum\limits_{t = 1}^T {\sum\limits_{i = 1}^n {({2{{p^2{G_1^2} }}{\alpha _t}} + {2p^2G_2^2{\alpha _t}{{\| {{q_{i,t + 1}}} \|}^2}} + \frac{{ {\mathbf{E}_{{\mathfrak{U}_t}}}[{{\| {\varepsilon _{i,t}^z} \|}^2}]}}{{4{\alpha _t}}})} }. \label{lemma5-eq2}
\end{flalign}

We have
\begin{flalign}
\sum\limits_{t = 1}^T {\sum\limits_{i = 1}^n {{G_2}\big( {R( \mathbb{X} ){\xi _t} + {\delta _t}} \big)\| {{q_{i,t + 1}}} \|} }  \le \sum\limits_{t = 1}^T {\sum\limits_{i = 1}^n {\big( {{2G_2^2{\alpha _t}{{\| {{q_{i,t + 1}}} \|}^2}} + \frac{{ R{{( \mathbb{X} )}^2}\xi _t^2}}{{4{\alpha _t}}} + \frac{{ \delta _t^2}}{{4{\alpha _t}}}} \big)} }. \label{lemma5-eq3}
\end{flalign}

From \eqref{Algorithm1-eq2}, for all $t \in {\mathbb{N}_ + }$, we have
\begin{flalign}
\| {{q_{i,t + 1}}} \| = {\gamma _t}\| {{{[ {{g_{i,t}}( {{x_{i,t}}} )} ]}_ + }} \|. \label{lemma5-eq4}
\end{flalign}

Then, we have
\begin{flalign}
{2(p^2+1)G_2^2{\alpha _t}{{\| {{q_{i,t + 1}}} \|}^2}} - q_{i,t + 1}^T{g_{i,t}}( {{x_{i,t}}} ) = ( {{2(p^2+1)G_2^2{\gamma _0}} - 1} ){\gamma _t}{\| {{{[ {{g_{i,t}}( {{x_{i,t}}} )} ]}_ + }} \|^2} \le 0. \label{lemma5-eq5}
\end{flalign}
where the equality holds due to the fact that $[ b ]_ + ^Tb = {\| {{{[ b ]}_ + }} \|^2}$ for any vector $b$ and the inequality holds due to ${\gamma _0} \le 1 /( {2(p^2+1)G_2^2} )$.

Combining \eqref{lemma5-eq1}--\eqref{lemma5-eq3} and \eqref{lemma5-eq5} yields \eqref{lemma5-eqa}.

\noindent (ii) From \eqref{ass2-eq1a}, we have
\begin{flalign}
{f_{i,t}}( y ) - {f_{i,t}}( {{x_{i,t}}} )  \le 2F, \forall y \in \mathbb{X}. \label{lemma5-2-eq1}
\end{flalign}

Dividing \eqref{lemma4-eq1} by ${{\gamma _t}}$, using \eqref{lemma5-2-eq1}, and summing over $t \in [ T ]$ gives
\begin{flalign}
\nonumber
&\;\;\;\;\;\sum\limits_{i = 1}^n {\sum\limits_{t = 1}^T {\frac{{q_{i,t + 1}^T{g_{i,t}}( {{x_{i,t}}} )}}{{{\gamma _t}}}} } \\
\nonumber
&\;\; \le {h_T}( y ) + \sum\limits_{t = 1}^T {\frac{{2nF}}{{{\gamma _t}}}}  + \sum\limits_{t = 1}^T {\frac{{{{\tilde \Delta }_t}}}{{{\gamma _t}}}}  + \sum\limits_{i = 1}^n {\sum\limits_{t = 1}^T {\frac{{{\Delta _{i,t}}( {\hat y} )}}{{{\gamma _t}}}} }  \\
&\;\; + \sum\limits_{i = 1}^n {\sum\limits_{t = 1}^T {\frac{{{G_2}\big( {R( \mathbb{X} ){\xi _t} + {\delta _t}} \big)\| {{q_{i,t + 1}}} \|}}{{{\gamma _t}}}} }  + \sum\limits_{t = 1}^T {\frac{{n{G_1}R( \mathbb{X} ){\xi _t} + n{G_1}{\delta _t}}}{{{\gamma _t}}}}. \label{lemma5-2-eq2}
\end{flalign}

We have
\begin{flalign}
\nonumber
&\;\;\;\;\; \sum\limits_{i = 1}^n {\sum\limits_{t = 1}^T {\frac{{( {p{G_1} + p{G_2}\| {{q_{i,t + 1}}} \|} ){\mathbf{E}_{{\mathfrak{U}_t}}}[\| {\varepsilon _{i,t}^z} \|]}}{{{\gamma _t}}}} } \\
&\;\; \le \sum\limits_{i = 1}^n {\sum\limits_{t = 1}^T {\big( {\frac{{2{ p^2{G_1^2} }{\gamma _0}}}{{\gamma _t^2}} + \frac{{2p^2G_2^2{\gamma _0}{{\| {{q_{i,t + 1}}} \|}^2}}}{{\gamma _t^2}} + \frac{{ {\mathbf{E}_{{\mathfrak{U}_t}}}[{{\| {\varepsilon _{i,t}^z} \|}^2}]}}{{4{\gamma _0}}}} \big)} }. \label{lemma5-2-eq3}
\end{flalign}

We have
\begin{flalign}
\sum\limits_{i = 1}^n {\sum\limits_{t = 1}^T {\frac{{{G_2}\big( {R( \mathbb{X} ){\xi _t} + {\delta _t}} \big)\| {{q_{i,t + 1}}} \|}}{{{\gamma _t}}}} } 
\le \sum\limits_{i = 1}^n {\sum\limits_{t = 1}^T {\big( {\frac{{2G_2^2{\gamma _0}{{\| {{q_{i,t + 1}}} \|}^2}}}{{\gamma _t^2}} + \frac{{R{{( \mathbb{X} )}^2}\xi _t^2}}{{4{\gamma _0}}} + \frac{{\delta _t^2}}{{4{\gamma _0}}}} \big)} }. \label{lemma5-2-eq4}
\end{flalign}

It follows from ${\gamma _t}{\alpha _t} = {\gamma _0}$ and \eqref{ass1-eq1} that
\begin{flalign}
\sum\limits_{t = 1}^T {\frac{{{\Delta _{i,t}}( {\hat y} )}}{{{\gamma _t}}}}  = \sum\limits_{t = 1}^T {\frac{1}{{{2\gamma _0}}}} {\mathbf{E}_{{\mathfrak{U}_t}}}[{{{\| {\hat y - {x_{i,t}}} \|}^2} - {{\| {\hat y - {x_{i,t + 1}}} \|}^2}}] 
\le \frac{{\mathbf{E}_{{\mathfrak{U}_t}}}[{{{\| {\hat y - {x_{i,1}}} \|}^2}}]}{{{2\gamma _0}}}. \label{lemma5-2-eq5}
\end{flalign}

From \eqref{Algorithm1-eq2} and the fact that $[ b ]_ + ^Tb = {\| {{{[ b ]}_ + }} \|^2}$ for any vector $b$, we have
\begin{flalign}
\frac{{q_{i,t + 1}^T{g_{i,t}}( {{x_{i,t}}} )}}{{{\gamma _t}}} = [ {{g_{i,t}}( {{x_{i,t}}} )} ]_ + ^T{g_{i,t}}( {{x_{i,t}}} ) = {\| [ {{g_{i,t}}( {{x_{i,t}}} )} ]_ + \|^2}. \label{lemma5-2-eq6}
\end{flalign}

Combining \eqref{lemma5-2-eq2}--\eqref{lemma5-2-eq6} and \eqref{lemma5-eq4}, and using $1/2 \ge 2(p^2+1)G_2^2{\gamma _0}$ yields \eqref{lemma5-eqb}.

\noindent (iii) From \eqref{Algorithm1-eq1}, ${\|  \cdot  \|}$ is convex, and $\sum\nolimits_{i = 1}^n {{{\left[ {{W_t}} \right]}_{ij}}}  = \sum\nolimits_{j = 1}^n {{{\left[ {{W_t}} \right]}_{ij}}}  = 1$, we have
\begin{flalign}
\nonumber
\frac{1}{n}\sum\limits_{i = 1}^n {\sum\limits_{j = 1}^n {\| {{x_{i,t}} - {x_{j,t}}} \|} }  &= \frac{1}{n}\sum\limits_{i = 1}^n {\sum\limits_{j = 1}^n {\Big\| {\sum\limits_{l = 1}^n {{{[ {{W_t}} ]}_{il}}{z_{l,t}}}  - {{\bar z}_t} + {{\bar z}_t} - \sum\limits_{l = 1}^n {{{[ {{W_t}} ]}_{j,l}}} {z_{l,t}}} \Big\|} } \\
\nonumber
& \le 2\sum\limits_{i = 1}^n {\Big\| {\sum\limits_{j = 1}^n {{{[ {{W_t}} ]}_{ij}}{z_{j,t}}}  - {{\bar z}_t}} \Big\|}  = 2\sum\limits_{i = 1}^n {\Big\| {\sum\limits_{j = 1}^n {{{[ {{W_t}} ]}_{ij}}( {{z_{j,t}} - {{\bar z}_t}} )} } \Big\|} \\
& \le 2\sum\limits_{i = 1}^n {\sum\limits_{j = 1}^n {{{[ {{W_t}} ]}_{ij}}\| {{z_{j,t}} - {{\bar z}_t}} \|} }  = 2\sum\limits_{i = 1}^n {\| {{z_{i,t}} - {{\bar z}_t}} \|}. \label{lemma5-3-eq1}
\end{flalign}

We have
\begin{flalign}
\sum\limits_{t = 3}^T {\sum\limits_{s = 1}^{t - 2} {{\lambda ^{t - s - 2}}} } \sum\limits_{j = 1}^n {\| {\varepsilon _{j,s}^z} \|}  = \sum\limits_{t = 1}^{T - 2} {\sum\limits_{j = 1}^n {\| {\varepsilon _{j,t}^z} \|} } \sum\limits_{s = 0}^{T - t - 2} {{\lambda ^s}}  \le \frac{1}{{1 - \lambda }}\sum\limits_{t = 1}^{T - 2} {\sum\limits_{j = 1}^n {\| {\varepsilon _{j,t}^z} \|} }. \label{lemma5-3-eq2}
\end{flalign}

We have
\begin{flalign}
\nonumber
&\;\;\;\;\; \frac{1}{n}\sum\limits_{t = 1}^T {\sum\limits_{i = 1}^n {\sum\limits_{j = 1}^n {\| {{x_{i,t}} - {x_{j,t}}} \|} } } \\
\nonumber
& \le \sum\limits_{t = 1}^T {\sum\limits_{i = 1}^n {2\| {{z_{i,t}} - {{\bar z}_t}} \|} } \\
\nonumber
& \le \sum\limits_{t = 1}^T {\sum\limits_{i = 1}^n {2\tau {\lambda ^{t - 2}}\sum\limits_{j = 1}^n {\| {{z_{j,1}}} \|} } }  + \sum\limits_{t = 2}^T {\sum\limits_{i = 1}^n {2\Big( {\frac{1}{n}\sum\limits_{j = 1}^n {\| {\varepsilon _{j,t - 1}^z} \|}  + \| {\varepsilon _{i,t - 1}^z} \|} \Big)} } \\
\nonumber
&\;\; + \sum\limits_{t = 3}^T {\sum\limits_{i = 1}^n {2\tau } \sum\limits_{s = 1}^{t - 2} {{\lambda ^{t - s - 2}}} \sum\limits_{j = 1}^n {\| {\varepsilon _{j,s}^z} \|} } \\
& \le n{\tilde{\varepsilon} _1} + \sum\limits_{t = 2}^T {\sum\limits_{i = 1}^n {4\| {\varepsilon _{i,t - 1}^z} \|} }  + \frac{{2n\tau }}{{1 - \lambda }}\sum\limits_{t = 1}^{T - 2} {\sum\limits_{j = 1}^n {\| {\varepsilon _{j,t}^z} \|} }, \label{lemma5-3-eq3}
\end{flalign}
where the first inequality holds due to \eqref{lemma5-3-eq1}; the second inequality holds due to \eqref{lemma3-eq1}; and the last inequality holds due to \eqref{lemma5-3-eq2}.\
It follows from \eqref{lemma5-3-eq3} that \eqref{lemma5-eqc} holds.

\noindent (iv) Similar to the way to get \eqref{lemma5-3-eq1}, from \eqref{Algorithm1-eq1}, ${\|  \cdot  \|^2}$ is convex, and $\sum\nolimits_{i = 1}^n {{{\left[ {{W_t}} \right]}_{ij}}}  = \sum\nolimits_{j = 1}^n {{{\left[ {{W_t}} \right]}_{ij}}}  = 1$, we have
\begin{flalign}
\frac{1}{n}\sum\limits_{i = 1}^n {\sum\limits_{j = 1}^n {{{\| {{x_{i,t}} - {x_{j,t}}} \|}^2}} }  \le \sum\limits_{i = 1}^n {4{{\| {{z_{i,t}} - {{\bar z}_t}} \|}^2}}. \label{lemma5-4-eq1}
\end{flalign}

From \eqref{lemma3-eq1}, we have
\begin{flalign}
\nonumber
&\;\;\;\;\; 4\sum\limits_{i = 1}^n {\sum\limits_{t = 1}^T {{{\| {{z_{i,t}} - {{\bar z}_t}} \|}^2}} } \\
\nonumber
& \le 4\sum\limits_{i = 1}^n {\sum\limits_{t = 1}^T {{\Big( {\tau {\lambda ^{t - 2}}\sum\limits_{j = 1}^n {\| {{z_{j,1}}} \|}  + \| {\varepsilon _{i,t - 1}^z} \| + \frac{1}{n}\sum\limits_{j = 1}^n {\| {\varepsilon _{j,t - 1}^z} \|}  + \tau \sum\limits_{s = 1}^{t - 2} {{\lambda ^{t - s - 2}}\sum\limits_{j = 1}^n {\| {\varepsilon _{j,s}^z} \|} } } \Big)^2}} } \\
\nonumber
& \le 16\sum\limits_{i = 1}^n {\sum\limits_{t = 1}^T {\Big( {{\Big( {\tau {\lambda ^{t - 2}}\sum\limits_{j = 1}^n {\| {{z_{j,1}}} \|} } \Big)^2} + {{\| {\varepsilon _{i,t - 1}^z} \|}^2} + {\Big( {\frac{1}{n}\sum\limits_{j = 1}^n {\| {\varepsilon _{j,t - 1}^z} \|} } \Big)^2} + {\Big( {\tau \sum\limits_{s = 1}^{t - 2} {{\lambda ^{t - s - 2}}\sum\limits_{j = 1}^n {\| {\varepsilon _{j,s}^z} \|} } } \Big)^2}} \Big)} } \\
\nonumber
& \le 16\sum\limits_{i = 1}^n {\sum\limits_{t = 1}^T {\Big( {{\Big( {\tau {\lambda ^{t - 2}}\sum\limits_{j = 1}^n {\| {{z_{j,1}}} \|} } \Big)^2} + {\| {\varepsilon _{i,t - 1}^z} \|^2} + \frac{1}{n}\sum\limits_{j = 1}^n {{\| {\varepsilon _{j,t - 1}^z} \|^2}}  + {\tau ^2}\sum\limits_{s = 1}^{t - 2} {{\lambda ^{t - s - 2}}\sum\limits_{s = 1}^{t - 2} {{\lambda ^{t - s - 2}}} {\Big( {\sum\limits_{j = 1}^n {\| {\varepsilon _{j,s}^z} \|} } \Big)^2}} } \Big)} } \\
\nonumber
& \le 16\sum\limits_{i = 1}^n {\sum\limits_{t = 1}^T {\Big( {{\Big( {\tau {\lambda ^{t - 2}}\sum\limits_{j = 1}^n {\| {{z_{j,1}}} \|} } \Big)^2} + 2{\| {\varepsilon _{i,t - 1}^z} \|^2} + \frac{{n{\tau ^2}}}{{1 - \lambda }}\sum\limits_{s = 1}^{t - 2} {{\lambda ^{t - s - 2}}} \sum\limits_{j = 1}^n {{\| {\varepsilon _{j,s}^z} \|^2}} } \Big)} } \\
\nonumber
& = 16\sum\limits_{i = 1}^n {\sum\limits_{t = 1}^T {\Big( {{\Big( {\tau {\lambda ^{t - 2}}\sum\limits_{j = 1}^n {\| {{z_{j,1}}} \|} } \Big)^2} + 2{\| {\varepsilon _{i,t - 1}^z} \|^2}} \Big)} }  + \frac{{16{n^2}{\tau ^2}}}{{1 - \lambda }}\sum\limits_{j = 1}^n {\sum\limits_{t = 1}^{T - 2} {{{\| {\varepsilon _{j,t}^z} \|}^2}\sum\limits_{s = 0}^{T - t - 2} {{\lambda ^s}} } } \\
& \le {{\tilde \varepsilon }_3} + {{\tilde \varepsilon }_4}\sum\limits_{t = 1}^T {\sum\limits_{i = 1}^n {{{\| {\varepsilon _{i,t}^z} \|}^2}} }, \label{lemma5-4-eq2}
\end{flalign}
where the third inequality holds due to the H\"{o}lder’s inequality. It follows form \eqref{lemma5-4-eq1}--\eqref{lemma5-4-eq2} that \eqref{lemma5-eqd} holds.

\noindent (v)
Applying \eqref{lemma1-eq1b} to update \eqref{Algorithm1-eq4} gives
\begin{flalign}
\nonumber
{\mathbf{E}_{{\mathfrak{U}_t}}}[\| {\varepsilon _{i,t}^z} \|] &= {\mathbf{E}_{{\mathfrak{U}_t}}}[\| {{z_{i,t + 1}} - {x_{i,t}}} \|] \le {{\alpha _t}{\mathbf{E}_{{\mathfrak{U}_t}}}[\| {{{\hat \omega }_{i,t + 1}}} \|]} \\
\nonumber
& = {{\alpha _t}{\mathbf{E}_{{\mathfrak{U}_t}}}[\| {\hat \partial {f_{i,t}}( {{x_{i,t}}} ) + \hat \partial {g_{i,t}}( {{x_{i,t}}} ){q_{i,t + 1}}} \|]} \\
\nonumber
& \le {{\alpha _t}}\big( {p{G_1} + p{G_2}{\gamma _t}\| {{{[ {{g_{i,t}}( {{x_{i,t}}} )} ]}_ + }} \|} \big) \\
& = {p{G_1}{\alpha _t} + p{G_2}{\gamma _0}\| {{{[ {{g_{i,t}}( {{x_{i,t}}} )} ]}_ + }} \|}. \label{lemma5-5-eq1}
\end{flalign}
where the second equality holds due to \eqref{Algorithm1-eq3}; the last inequality holds due to \eqref{lemma2-eq1c}, \eqref{lemma2-eq1f} and \eqref{lemma5-eq4}; the last equality holds due to ${\gamma _t}{\alpha _t} = {\gamma _0}$. It follows form \eqref{lemma5-5-eq1} that \eqref{lemma5-eqe} holds.
\end{proof}

Next, network regret and cumulative constraint violation bounds for the general cases are provided in the following lemma.
\begin{lemma}\label{lem6}
Under the same condition as stated in Lemma~5, for any $T \in {\mathbb{N}_ + }$, it holds that
\begin{flalign}
& {\mathbf{E}}\Big[\frac{1}{n}\sum\limits_{i = 1}^n {\sum\limits_{t = 1}^T {{f_t}( {{x_{i,t}}} )} }  - \sum\limits_{t = 1}^T {{f_t}( y )}\Big]
\le {\varepsilon _1} {{G_1}}  + {n_T} + \frac{1}{n}\sum\limits_{i = 1}^n {\sum\limits_{t = 1}^T {\mathbf{E}}[{{\Delta _{i,t}}( {\hat y} )}] }, \label{lemma6-eqa}\\
& {\mathbf{E}}\Big[\frac{1}{n}\sum\limits_{i = 1}^n {\sum\limits_{t = 1}^T {\| {{{[ {{g_t}( {{x_{i,t}}} )} ]}_ + }} \|} }\Big]  \le \sqrt {2G_2^2{{\tilde \varepsilon }_3}T + {\varepsilon _2}T{\mathbf{E}}[{{\tilde h}_T}( {\hat y} )}], \label{lemma6-eqb}\\
& {\mathbf{E}}\Big[\frac{1}{n}\sum\limits_{i = 1}^n {\sum\limits_{t = 1}^T {\| {{{[ {{g_t}( {{x_{i,t}}} )} ]}_ + }} \|} }\Big]  \le n{G_2}{\tilde{\varepsilon} _1} + {\varepsilon _3}\sum\limits_{t = 1}^T {{\alpha _t}}  + {\varepsilon _4}\sum\limits_{t = 1}^T {\sum\limits_{i = 1}^n {\| {{{[ {{g_{i,t}}( {{x_{i,t}}} )} ]}_ + }} \|} }, \label{lemma6-eqc}
\end{flalign}
where
\begin{flalign}
\nonumber
&{n_T} = \sum\limits_{t = 1}^T {{\tilde \varepsilon _2^2{{ {{G_1^2} } }}{\alpha _t}}}  + {m_T}, {\varepsilon _2} = \frac{{4\max \{ {1,G_2^2{{\tilde \varepsilon }_4}} \}}}{{\min \{ {1,\frac{1 }{{2{\gamma _0}}}} \}}}, \\
\nonumber
&{\varepsilon _3} = {np{{G_1} }{G_2}{{\tilde \varepsilon }_2}}, {\varepsilon _4} = {pG_2^2{{\tilde \varepsilon }_2}{\gamma _0} + 1 }.
\end{flalign}
\end{lemma}
\begin{proof}
We will show that \eqref{lemma6-eqa}--\eqref{lemma6-eqc} hold in the following, respectively.

\noindent (i)
From ${f_t}( x ) = \frac{1}{n}\sum\nolimits_{j = 1}^n {{f_{j,t}}( x )} $, we have
\begin{flalign}
\nonumber
\sum\limits_{i = 1}^n {{f_t}( {{x_{i,t}}} )}  &= \frac{1}{n}\sum\limits_{i = 1}^n {\sum\limits_{j = 1}^n {{f_{j,t}}( {{x_{i,t}}} )} }  = \frac{1}{n}\sum\limits_{i = 1}^n {\sum\limits_{j = 1}^n {{f_{j,t}}( {{x_{j,t}}} )} }  + \frac{1}{n}\sum\limits_{i = 1}^n {\sum\limits_{j = 1}^n {\big( {{f_{j,t}}( {{x_{i,t}}} ) - {f_{j,t}}( {{x_{j,t}}} )} \big)} } \\
\nonumber
& = \sum\limits_{i = 1}^n {{f_{i,t}}( {{x_{i,t}}} )}  + \frac{1}{n}\sum\limits_{i = 1}^n {\sum\limits_{j = 1}^n {\big( {{f_{j,t}}( {{x_{i,t}}} ) - {f_{j,t}}( {{x_{j,t}}} )} \big)} } \\
& \le \sum\limits_{i = 1}^n {{f_{i,t}}( {{x_{i,t}}} )}  + \frac{1}{n}\sum\limits_{i = 1}^n {\sum\limits_{j = 1}^n {G_1\| {{x_{i,t}} - {x_{j,t}}} \|} }, \label{lemma6-1-eq1}
\end{flalign}
where the inequality holds due to \eqref{ass4-eq2a}.

From \eqref{lemma5-eqc}, we have
\begin{flalign}
\frac{1}{n}\sum\limits_{t = 1}^T {\sum\limits_{i = 1}^n {\sum\limits_{j = 1}^n {{G_1}\| {{x_{i,t}} - {x_{j,t}}} \|} } } \le n{\varepsilon _1} {{G_1} } + \sum\limits_{t = 1}^T {\sum\limits_{i = 1}^n {\Big( {{\tilde \varepsilon _2^2{{ {{G_1^2}} }}{\alpha _t}} + \frac{{ {{\| {\varepsilon _{i,t}^z} \|}^2}}}{{4{\alpha _t}}}} \Big)} }. \label{lemma6-1-eq2}
\end{flalign}

Combining \eqref{lemma6-1-eq1}--\eqref{lemma6-1-eq2} and \eqref{lemma5-eqa} yields \eqref{lemma6-eqa}.

\noindent (ii)
We have
\begin{flalign}
\nonumber
{\| {{{[ {{g_{i,t}}( {{x_{i,t}}} )} ]}_ + }} \|^2} &= {\| {{{[ {{g_{i,t}}( {{x_{i,t}}} )} ]}_ + } - {{[ {{g_{i,t}}( {{x_{j,t}}} )} ]}_ + } + {{[ {{g_{i,t}}( {{x_{j,t}}} )} ]}_ + }} \|^2} \\
\nonumber
& \ge \frac{1}{2}{\| {{{[ {{g_{i,t}}( {{x_{j,t}}} )} ]}_ + }} \|^2} - {\| {{{[ {{g_{i,t}}( {{x_{i,t}}} )} ]}_ + } - {{[ {{g_{i,t}}( {{x_{j,t}}} )} ]}_ + }} \|^2} \\
\nonumber
& \ge \frac{1}{2}{\| {{{[ {{g_{i,t}}( {{x_{j,t}}} )} ]}_ + }} \|^2} - {\| {{g_{i,t}}( {{x_{i,t}}} ) - {g_{i,t}}( {{x_{j,t}}} )} \|^2} \\
& \ge \frac{1}{2}{\| {{{[ {{g_{i,t}}( {{x_{j,t}}} )} ]}_ + }} \|^2} - G_2^2{\| {{x_{i,t}} - {x_{j,t}}} \|^2}, \label{lemma6-2-eq1}
\end{flalign}
where the second and the third inequalities hold due to the nonexpansive property of the projection ${[  \cdot  ]_ + }$ and \eqref{ass4-eq2c}, respectively.

From ${g_t}( x ) = {\rm{col}}\big( {{g_{1,t}}( x ), \cdot  \cdot  \cdot ,{g_{n,t}}( x )} \big)$, we have
\begin{flalign}
\sum\limits_{t = 1}^T {\sum\limits_{i = 1}^n {\sum\limits_{j = 1}^n {{{\| {{{[ {{g_{i,t}}( {{x_{j,t}}} )} ]}_ + }} \|}^2}} } }  = \sum\limits_{t = 1}^T {\sum\limits_{j = 1}^n {{{\| {{{[ {{g_t}( {{x_{j,t}}} )} ]}_ + }} \|}^2}} }. \label{lemma6-2-eq2}
\end{flalign}

Summing \eqref{lemma6-2-eq1} over $i,j \in [ n ],t \in [ T ]$, dividing by $n$, and using \eqref{lemma6-2-eq2} and \eqref{lemma5-eqd} gives
\begin{flalign}
\frac{1}{n}\sum\limits_{t = 1}^T {\sum\limits_{j = 1}^n {{{\| {{{[ {{g_t}( {{x_{j,t}}} )} ]}_ + }} \|}^2}} }  \le 2G_2^2{{\tilde \varepsilon }_3} + \sum\limits_{t = 1}^T {\sum\limits_{i = 1}^n {2\big( {{{\| {{{[ {{g_{i,t}}( {{x_{i,t}}} )} ]}_ + }} \|}^2} + G_2^2{{\tilde \varepsilon }_4}{{\| {\varepsilon _{i,t}^z} \|}^2}} \big)} }. \label{lemma6-2-eq3}
\end{flalign}

From ${g_{i,t}}( y ) \le {\mathbf{0}_{{m_i}}}, \forall i \in [ n ], \forall t \in {\mathbb{N}_ + }$, when $y \in {\mathcal{X}_T}$, we have
\begin{flalign}
{h_T}\left( y \right) \le 0. \label{lemma6-2-eq4}
\end{flalign}

Combining \eqref{lemma6-2-eq3}, \eqref{lemma6-2-eq4} and \eqref{lemma5-eqb} gives
\begin{flalign}
{\mathbf{E}}\Big[\frac{1}{n}\sum\limits_{i = 1}^n {\sum\limits_{t = 1}^T {{{\| {{{[ {{g_t}( {{x_{i,t}}} )} ]}_ + }} \|}^2}} }\Big]  \le 2G_2^2{{\tilde \varepsilon }_3} + {\varepsilon _2}{\mathbf{E}}[{{\tilde h}_T}( {\hat y} )], \forall y \in {\mathcal{X}_T}. \label{lemma6-2-eq5}
\end{flalign}

We have
\begin{flalign}
{\Big( {\frac{1}{n}\sum\limits_{i = 1}^n {\sum\limits_{t = 1}^T {\| {{{[ {{g_t}( {{x_{i,t}}} )} ]}_ + }} \|} } } \Big)^2} \le \frac{T}{n}\sum\limits_{i = 1}^n {\sum\limits_{t = 1}^T {{{\| {{{[ {{g_t}( {{x_{i,t}}} )} ]}_ + }} \|}^2}} }. \label{lemma6-2-eq6}
\end{flalign}

Combining \eqref{lemma6-2-eq5} and \eqref{lemma6-2-eq6} yields \eqref{lemma6-eqb}.

\noindent (iii)
From \eqref{lemma5-eqc} and \eqref{lemma5-eqe}, we have
\begin{flalign}
\nonumber
{\mathbf{E}}\Big[\frac{1}{n}\sum\limits_{t = 1}^T {\sum\limits_{i = 1}^n {\sum\limits_{j = 1}^n {\| {{x_{i,t}} - {x_{j,t}}} \|} } }\Big]  &\le n{\tilde{\varepsilon} _1} + {{\tilde \varepsilon }_2}\sum\limits_{t = 1}^T {\sum\limits_{i = 1}^n {\mathbf{E}_{{\mathfrak{U}_t}}}[{\| {\varepsilon _{i,t}^z} \|}] } \\
& \le n{\tilde{\varepsilon} _1} + {{{\tilde \varepsilon }_2}}\sum\limits_{t = 1}^T {\sum\limits_{i = 1}^n {\big( {p{G_1}{\alpha _t} + p{G_2}{\gamma _0}\| {{{[ {{g_{i,t}}( {{x_{i,t}}} )} ]}_ + }} \|} \big)} }. \label{lemma6-3-eq1}
\end{flalign}

From \eqref{ass4-eq2c}, we have
\begin{flalign}
\nonumber
&\;\;\;\;\; \frac{1}{n}\sum\limits_{j = 1}^n {\sum\limits_{t = 1}^T {\| {{{[ {{g_t}( {{x_{j,t}}} )} ]}_ + }} \|} }  \le \frac{1}{n}\sum\limits_{i = 1}^n {\sum\limits_{j = 1}^n {\sum\limits_{t = 1}^T {\| {{{[ {{g_{i,t}}( {{x_{j,t}}} )} ]}_ + }} \|} } } \\
\nonumber
& = \frac{1}{n}\sum\limits_{t = 1}^T {\sum\limits_{i = 1}^n {\sum\limits_{j = 1}^n {\| {{{[ {{g_{i,t}}( {{x_{i,t}}} )} ]}_ + } + {{[ {{g_{i,t}}( {{x_{j,t}}} )} ]}_ + } - {{[ {{g_{i,t}}( {{x_{i,t}}} )} ]}_ + }} \|} } } \\
& \le \frac{1}{n}\sum\limits_{t = 1}^T {\sum\limits_{i = 1}^n {\sum\limits_{j = 1}^n {\big( {\| {{{[ {{g_{i,t}}( {{x_{i,t}}} )} ]}_ + }} \| + {G_2}\| {{x_{i,t}} - {x_{j,t}}} \|} \big)} } }. \label{lemma6-3-eq2}
\end{flalign}

Combining \eqref{lemma6-3-eq1} and \eqref{lemma6-3-eq2} yields \eqref{lemma6-eqc}.
\end{proof}

\hspace{-3mm}\emph{B. Proof of Theorem~1}

We will show that \eqref{theorem1-eq2} and \eqref{theorem1-eq3} in Theorem~1 hold in the following, respectively.

\noindent (i)
From \eqref{ass1-eq1}, we have
\begin{flalign}
\nonumber
&\;\;\sum\limits_{t = 1}^T {\frac{1}{{{\alpha _t}}}} \big( {{{\| {\hat y - {x_{i,t}}} \|}^2} - {{\| {\hat y - {x_{i,t + 1}}} \|}^2}} \big)\\
\nonumber
& = \sum\limits_{t = 1}^T {\Big( {\frac{1}{{{\alpha _{t - 1}}}}{{\| {\hat y - {x_{i,t}}} \|}^2} - \frac{1}{{{\alpha _t}}}{{\| {\hat y - {x_{i,t + 1}}} \|}^2} + \Big( {\frac{1}{{{\alpha _t}}} - \frac{1}{{{\alpha _{t - 1}}}}} \Big){{\| {\hat y - {x_{i,t}}} \|}^2}} \Big)} \\
& \le \frac{1}{{{\alpha _0}}}{\| {\hat y - {x_{i,1}}} \|^2} - \frac{1}{{{\alpha _T}}}{\| {\hat y - {x_{i,T + 1}}} \|^2} + \Big( {\frac{1}{{{\alpha _T}}} - \frac{1}{{{\alpha _0}}}} \Big)4R{( \mathbb{X} )^2} \le \frac{{4R{{( \mathbb{X} )}^2}}}{{{\alpha _T}}}. \label{theorem1-1-eq1}
\end{flalign}

From \eqref{theorem1-eq1}, we have
\begin{flalign}
\sum\limits_{t = 1}^T {{\alpha _t}}  &= \sum\limits_{t = 2}^T {\frac{1}{{{t^c}}}}  + 1 \le \int_1^T {\frac{1}{{{t^c}}}} dt + 1 \le \frac{{{T^{1 - c}}}}{{1 - c}}, \label{theorem1-1-eq2} \\
\sum\limits_{t = 1}^T {\frac{{\xi _t^2}}{{{\alpha _t}}}} &= \sum\limits_{t = 1}^T {{\alpha _t}}  \le \frac{{{T^{1 - c}}}}{{1 - c}}, \label{theorem1-1-eq3}\\
\sum\limits_{t = 1}^T {\frac{{\delta _t^2}}{{{\alpha _t}}}}  &= r{( \mathbb{X} )^2}\sum\limits_{t = 1}^T {{\alpha _t}}  \le \frac{{r{{( \mathbb{X} )}^2}{T^{1 - c}}}}{{1 - c}}, \label{theorem1-1-eq4}\\
\sum\limits_{t = 1}^T {{\xi _t}}  &= \sum\limits_{t = 1}^T {{\alpha _t}}  \le \frac{{{T^{1 - c}}}}{{1 - c}}, \label{theorem1-1-eq5}\\
\sum\limits_{t = 1}^T {{\delta _t}}  &= r{( \mathbb{X} )}\sum\limits_{t = 1}^T {{\alpha _t}}  \le \frac{{r{{( \mathbb{X} )}}{T^{1 - c}}}}{{1 - c}}, \label{theorem1-1-eq6}\\
\frac{{2R{{( \mathbb{X} )}^2}}}{{{\alpha _T}}} &= 2R{{( \mathbb{X} )}^2}{T^c}. \label{theorem1-1-eq7}
\end{flalign}

Denote
\begin{flalign}
\nonumber
{x^*} = \mathop {\arg \min }\limits_{x \in {\mathcal{X}_T}} \sum\limits_{t = 1}^T {{l_t}( x )}.
\end{flalign}

Choosing $y = {x^*}$, and combining \eqref{lemma6-eqa} and \eqref{theorem1-1-eq1}--\eqref{theorem1-1-eq7} gives
\begin{flalign}
\nonumber
{\mathbf{E}}[ {{\rm{Net} \mbox{-} \rm{Reg}}( T )} ] &\le {\varepsilon _1} {{G_1}} + \frac{{{{ {{G_1^2}} }}\tilde \varepsilon _2^2{T^{1 - c}}}}{{ ( {1 - c} )}} + \frac{{2{{( {p^2{G_1^2} } )}}{T^{1 - c}}}}{{ ( {1 - c} )}} + \frac{{ R{{( \mathbb{X} )}^2}{T^{1 - c}}}}{{4( {1 - c} )}} \\
&\;\; + \frac{{ r{(\mathbb{X})^2}{T^{1 - c}}}}{{4( {1 - c} )}} + \frac{{ {{G_1}} R( \mathbb{X} ){T^{1 - c}}}}{{1 - c}} + \frac{{{G_1}r( \mathbb{X} ){T^{1 - c}}}}{{1 - c}} + 2R{{( \mathbb{X} )}^2}{T^c}, \label{theorem1-1-eq8}
\end{flalign}
which yields \eqref{theorem1-eq2}.

\noindent (ii)
From \eqref{ass1-eq1}, we have
\begin{flalign}
\sum\limits_{i = 1}^n \frac{{\mathbf{E}_{{\mathfrak{U}_t}}}[{{{\| {\hat y - {x_{i,1}}} \|}^2}}]}{{{2\gamma _0}}} \le \frac{{2nR{{( \mathbb{X} )}^2}}}{{{\gamma _0}}}. \label{theorem1-2-eq1}
\end{flalign}

From \eqref{theorem1-eq1}, we have
\begin{flalign}
\sum\limits_{t = 1}^T {\frac{1}{{{\gamma _t}}}} = \sum\limits_{t = 1}^T {\frac{{{\alpha _t}}}{{{\gamma _0}}}}  \le \frac{{{T^{1 - c}}}}{{{\gamma _0}( {1 - c} )}}. \label{theorem1-2-eq2}
\end{flalign}

When $c \in ( {0,1/2} )$, from \eqref{theorem1-eq1}, we have
\begin{flalign}
\nonumber
&\sum\limits_{t = 1}^T {\alpha _t^2}  \le \frac{{{T^{1 - 2c}}}}{{1 - 2c}}, \\
&\sum\limits_{t = 1}^T {\frac{1}{{\gamma _t^2}}}  = \sum\limits_{t = 1}^T {\frac{{\alpha _t^2}}{{\gamma _0^2}}}  \le \frac{{{T^{1 - 2c}}}}{{\gamma _0^2( {1 - 2c} )}}, \label{theorem1-2-eq3}\\
&\sum\limits_{t = 1}^T {\xi _t^2}  = \sum\limits_{t = 1}^T {\alpha _t^2}  \le \frac{{{T^{1 - 2c}}}}{{1 - 2c}}, \label{theorem1-2-eq4}\\
&\sum\limits_{t = 1}^T {\delta _t^2}  = r{(\mathbb{X})^2}\sum\limits_{t = 1}^T {\alpha _t^2}  \le \frac{{r{{(\mathbb{X})}^2}{T^{1 - 2c}}}}{{1 - 2c}}, \label{theorem1-2-eq5}\\
&\sum\limits_{t = 1}^T {\frac{{{\xi _t}}}{{{\gamma _t}}}}  = \frac{1}{{{\gamma _0}}}\sum\limits_{t = 1}^T {\alpha _t^2}  \le \frac{{{T^{1 - 2c}}}}{{{\gamma _0}( {1 - 2c} )}}, \label{theorem1-2-eq6}\\
&\sum\limits_{t = 1}^T {\frac{{{\delta _t}}}{{{\gamma _t}}}}  = \frac{{r(\mathbb{X})}}{{{\gamma _0}}}\sum\limits_{t = 1}^T {\alpha _t^2}  \le \frac{{r(\mathbb{X}){T^{1 - 2c}}}}{{{\gamma _0}( {1 - 2c} )}}. \label{theorem1-2-eq7}
\end{flalign}

Combining \eqref{lemma6-eqb} and \eqref{theorem1-2-eq1}--\eqref{theorem1-2-eq7} yields
\begin{flalign}
\nonumber
&\;\;\;\;\;{\mathbf{E}}\Big[{\Big( {\frac{1}{n}\sum\limits_{i = 1}^n {\sum\limits_{t = 1}^T {\| {{{[ {{g_t}( {{x_{i,t}}} )} ]}_ + }} \|} } } \Big)^2}\Big] \\
\nonumber
&\le 2G_2^2{{\tilde \varepsilon }_3}T + \frac{{2n{R{{( \mathbb{X} )}^2}}{\varepsilon _2}T}}{{{\gamma _0}}} + \frac{{2nF{\varepsilon _2}{T^{2 - c}}}}{{{\gamma _0}( {1 - c} )}} + \frac{{2n{ {p^2{G_1^2} }}{\varepsilon _2}{T^{2 - 2c}}}}{{ {\gamma _0}( {1 - 2c} )}} + \frac{{n R{{( \mathbb{X} )}^2}{\varepsilon _2}{T^{2 - 2c}}}}{{4{\gamma _0}( {1 - 2c} )}}\\
&\;\; + \frac{{n r{{( \mathbb{X} )}^2}{\varepsilon _2}{T^{2 - 2c}}}}{{4{\gamma _0}( {1 - 2c} )}} + \frac{{n{{G_1} }R( \mathbb{X} ){\varepsilon _2}{T^{2 - 2c}}}}{{{\gamma _0}( {1 - 2c} )}} + \frac{{n{G_1}r( \mathbb{X} ){\varepsilon _2}{T^{2 - 2c}}}}{{{\gamma _0}( {1 - 2c} )}}. \label{theorem1-2-eq8}
\end{flalign}

When $c = 1/2$, from \eqref{theorem1-eq1}, we have
\begin{flalign}
\nonumber
&\sum\limits_{t = 1}^T {\alpha _t^2}  \le 2\log ( T ), \\
&\sum\limits_{t = 1}^T {\frac{1}{{\gamma _t^2}}}  = \sum\limits_{t = 1}^T {\frac{{\alpha _t^2}}{{\gamma _0^2}}}  \le \frac{{2\log ( T )}}{{\gamma _0^2}}, \label{theorem1-2-eq9}\\
&\sum\limits_{t = 1}^T {\xi _t^2}  = \sum\limits_{t = 1}^T {\alpha _t^2}  \le 2\log ( T ), \label{theorem1-2-eq10}\\
&\sum\limits_{t = 1}^T {\delta _t^2}  = r{(\mathbb{X})^2}\sum\limits_{t = 1}^T {\alpha _t^2}  \le 2r{(\mathbb{X})^2}\log ( T ), \label{theorem1-2-eq11}\\
&\sum\limits_{t = 1}^T {\frac{{{\xi _t}}}{{{\gamma _t}}}}  = \frac{1}{{{\gamma _0}}}\sum\limits_{t = 1}^T {\alpha _t^2}  \le \frac{{2\log ( T )}}{{{\gamma _0}}}, \label{theorem1-2-eq12}\\
&\sum\limits_{t = 1}^T {\frac{{{\delta _t}}}{{{\gamma _t}}}}  = \frac{{r(\mathbb{X})}}{{{\gamma _0}}}\sum\limits_{t = 1}^T {\alpha _t^2}  \le \frac{{2r(\mathbb{X})\log ( T )}}{{{\gamma _0}}}. \label{theorem1-2-eq13}
\end{flalign}

Combining \eqref{lemma6-eqb}, \eqref{theorem1-2-eq1}, \eqref{theorem1-2-eq2}, and \eqref{theorem1-2-eq9}--\eqref{theorem1-2-eq13} yields
\begin{flalign}
\nonumber
&\;\;\;\;\;{\mathbf{E}}\Big[{\Big( {\frac{1}{n}\sum\limits_{i = 1}^n {\sum\limits_{t = 1}^T {\| {{{[ {{g_t}( {{x_{i,t}}} )} ]}_ + }} \|} } } \Big)^2}\Big] \\
\nonumber
&\le 2G_2^2{{\tilde \varepsilon }_3}T + \frac{{2nR{( \mathbb{X} )^2}{\varepsilon _2}T}}{{{\gamma _0}}} + \frac{{2nF{\varepsilon _2}{T^{2 - c}}}}{{{\gamma _0}( {1 - c} )}} + \frac{{4n{{ {p^2{G_1^2} } }}{\varepsilon _2}T\log ( T )}}{{ {\gamma _0}}} + \frac{{n R{{( \mathbb{X} )}^2}{\varepsilon _2}T\log ( T )}}{{2{\gamma _0}}}\\
&\;\; + \frac{{n r{{( \mathbb{X} )}^2}{\varepsilon _2}T\log ( T )}}{{2{\gamma _0}}} + \frac{{2n {{G_1} } R( \mathbb{X} ){\varepsilon _2}T\log ( T )}}{{{\gamma _0}}} + \frac{{2n{G_1}r( \mathbb{X} ){\varepsilon _2}T\log ( T )}}{{{\gamma _0}}}. \label{theorem1-2-eq14}
\end{flalign}

When $c = ( {1/2,1} )$, from \eqref{theorem1-eq1}, there exists a constant $Q > 0$ such that
\begin{flalign}
\nonumber
&\sum\limits_{t = 1}^T {\alpha _t^2}  \le Q, \\
&\sum\limits_{t = 1}^T {\frac{1}{{\gamma _t^2}}}  = \sum\limits_{t = 1}^T {\frac{{\alpha _t^2}}{{\gamma _0^2}}}  \le \frac{Q}{{\gamma _0^2}}, \label{theorem1-2-eq15}\\
&\sum\limits_{t = 1}^T {\xi _t^2}  = \sum\limits_{t = 1}^T {\alpha _t^2}  \le Q, \label{theorem1-2-eq16}\\
&\sum\limits_{t = 1}^T {\delta _t^2}  = r{(\mathbb{X})^2}\sum\limits_{t = 1}^T {\alpha _t^2}  \le r{(\mathbb{X})^2}Q, \label{theorem1-2-eq17}\\
&\sum\limits_{t = 1}^T {\frac{{{\xi _t}}}{{{\gamma _t}}}}  = \frac{1}{{{\gamma _0}}}\sum\limits_{t = 1}^T {\alpha _t^2}  \le \frac{Q}{{{\gamma _0}}}, \label{theorem1-2-eq18}\\
&\sum\limits_{t = 1}^T {\frac{{{\delta _t}}}{{{\gamma _t}}}}  = \frac{{r(\mathbb{X})}}{{{\gamma _0}}}\sum\limits_{t = 1}^T {\alpha _t^2}  \le \frac{{r(\mathbb{X})Q}}{{{\gamma _0}}}. \label{theorem1-2-eq19}
\end{flalign}

Combining \eqref{lemma6-eqb}, \eqref{theorem1-2-eq1}, \eqref{theorem1-2-eq2}, and \eqref{theorem1-2-eq15}--\eqref{theorem1-2-eq19} yields
\begin{flalign}
\nonumber
&\;\;\;\;\;{\mathbf{E}}\Big[{\Big( {\frac{1}{n}\sum\limits_{i = 1}^n {\sum\limits_{t = 1}^T {\| {{{[ {{g_t}( {{x_{i,t}}} )} ]}_ + }} \|} } } \Big)^2}\Big] \\
\nonumber
&\le 2G_2^2{{\tilde \varepsilon }_3}T + \frac{{2n{R{{( \mathbb{X} )}^2}}{\varepsilon _2}T}}{{{\gamma _0}}} + \frac{{2nF{\varepsilon _2}{T^{2 - c}}}}{{{\gamma _0}( {1 - c} )}} + \frac{{2n{{ {p^2{G_1^2} } }}{\varepsilon _2}QT}}{{ {\gamma _0}}} + \frac{{n R{{( \mathbb{X} )}^2}{\varepsilon _2}QT}}{{4{\gamma _0}}}\\
&\;\; + \frac{{n r{{( \mathbb{X} )}^2}{\varepsilon _2}QT}}{{4{\gamma _0}}} + \frac{{n {{G_1} } R( \mathbb{X} ){\varepsilon _2}QT}}{{{\gamma _0}}} + \frac{{n{G_1}r( \mathbb{X} ){\varepsilon _2}QT}}{{{\gamma _0}}}. \label{theorem1-2-eq20}
\end{flalign}

It follows form \eqref{theorem1-2-eq8}, \eqref{theorem1-2-eq14}, and \eqref{theorem1-2-eq20} that \eqref{theorem1-eq3} holds.

\hspace{-3mm}\emph{C. Proof of Theorem~2}

We will show that \eqref{theorem2-eq1} and \eqref{theorem2-eq2} in Theorem~2 hold in the following, respectively.

\noindent (i)
From \eqref{theorem1-1-eq8}, we have \eqref{theorem2-eq1}.

\noindent (ii)
When Slater’s condition holds, we have
\begin{flalign}
\nonumber
{h_T}( {{x_s}} ) &= \sum\limits_{i = 1}^n {\sum\limits_{t = 1}^T {\frac{{q_{i,t + 1}^T{g_{i,t}}( {{x_s}} )}}{{{\gamma _t}}}} }  = \sum\limits_{i = 1}^n {\sum\limits_{t = 1}^T {[ {{g_{i,t}}( {{x_{i,t}}} )} ]_ + ^T{g_{i,t}}( {{x_s}} )} } \\
\nonumber
& \le  - \sum\limits_{i = 1}^n {\sum\limits_{t = 1}^T {{\varsigma _s}[ {{g_{i,t}}( {{x_{i,t}}} )} ]_ + ^T{1_{{m_i}}}} }
=  - {\varsigma _s}{\sum\limits_{i = 1}^n \sum\limits_{t = 1}^T \| {{{[ {{g_{i,t}}( {{x_{i,t}}} )} ]}_ + }} \| _1} \\
& \le  - {\varsigma _s}\sum\limits_{i = 1}^n {\sum\limits_{t = 1}^T {\| {{{[ {{g_{i,t}}( {{x_{i,t}}} )} ]}_ + }} \|} }, \label{theorem2-2-eq1}
\end{flalign}
where the second equality holds due to \eqref{Algorithm1-eq2} and the first inequality holds due to \eqref{ass8-eq1}.

Choosing $y = {x_s}$ in \eqref{lemma5-eqb}, and using \eqref{Algorithm1-eq2} and \eqref{theorem2-2-eq1} gives
\begin{flalign}
{\varsigma _s}\sum\limits_{t = 1}^T {\sum\limits_{i = 1}^n {\| {{{[ {{g_{i,t}}( {{x_{i,t}}} )} ]}_ + }} \|} }  \le {{\tilde h}_T}( {{\hat y}} ). \label{theorem2-2-eq2}
\end{flalign}

When $c \in ( {0,1/2} )$, combining \eqref{theorem2-2-eq2}, \eqref{theorem1-eq1} and \eqref{theorem1-2-eq1}--\eqref{theorem1-2-eq7} gives
\begin{flalign}
\nonumber
&\;\;\;\;\;{\mathbf{E}}\Big[\sum\limits_{i = 1}^n {\sum\limits_{t = 1}^T {\| {{{[ {{g_{i,t}}( {{x_{i,t}}} )} ]}_ + }} \|} }\Big] \\
\nonumber
& \le \frac{{2n{R{{( \mathbb{X} )}^2}}}}{{{\varsigma _s}{\gamma _0}}} + \frac{{2nF{T^{1 - c}}}}{{{\varsigma _s}{\gamma _0}( {1 - c} )}} + \frac{{2n{{ {p^2{G_1^2}} }}{T^{1 - 2c}}}}{{ {\varsigma _s}{\gamma _0}( {1 - 2c} )}} + \frac{{n R{{( \mathbb{X} )}^2}{T^{1 - 2c}}}}{{4{\varsigma _s}{\gamma _0}( {1 - 2c} )}} + \frac{{n r{{( \mathbb{X} )}^2}{T^{1 - 2c}}}}{{4{\varsigma _s}{\gamma _0}( {1 - 2c} )}}\\
&\;\; + \frac{{n {{G_1}} R( \mathbb{X} ){T^{1 - 2c}}}}{{{\varsigma _s}{\gamma _0}( {1 - 2c} )}} + \frac{{n{G_1}r( \mathbb{X} ){T^{1 - 2c}}}}{{{\varsigma _s}{\gamma _0}( {1 - 2c} )}}. \label{theorem2-2-eq3}
\end{flalign}

From \eqref{theorem1-eq1}, \eqref{lemma6-eqc}, \eqref{theorem2-2-eq3}, we have
\begin{flalign}
\nonumber
&\;\;\;\;\;{\mathbf{E}}[{\rm{Net}\mbox{-}\rm{CCV}}( T )] \\
\nonumber
& \le n{G_2}{\tilde{\varepsilon} _1} + \frac{{{\varepsilon _3}{T^{1 - c}}}}{{1 - c}} + \frac{{2n{R{{( \mathbb{X} )}^2}}{\varepsilon _4}}}{{{\varsigma _s}{\gamma _0}}} + \frac{{2nF{\varepsilon _4}{T^{1 - c}}}}{{{\varsigma _s}{\gamma _0}( {1 - c} )}} + \frac{{2n{{ {p^2{G_1^2} } }}{\varepsilon _4}{T^{1 - 2c}}}}{{ {\varsigma _s}{\gamma _0}( {1 - 2c} )}} + \frac{{n R{{( \mathbb{X} )}^2}{\varepsilon _4}{T^{1 - 2c}}}}{{4{\varsigma _s}{\gamma _0}( {1 - 2c} )}}  \\
&\;\; + \frac{{n r{{( \mathbb{X} )}^2}{\varepsilon _4}{T^{1 - 2c}}}}{{4{\varsigma _s}{\gamma _0}( {1 - 2c} )}} + \frac{{n {{G_1} } R( \mathbb{X} ){\varepsilon _4}{T^{1 - 2c}}}}{{{\varsigma _s}{\gamma _0}( {1 - 2c} )}} + \frac{{n{G_1}r( \mathbb{X} ){\varepsilon _4}{T^{1 - 2c}}}}{{{\varsigma _s}{\gamma _0}( {1 - 2c} )}}. \label{theorem2-2-eq4}
\end{flalign}

When $c = 1/2$, combining \eqref{theorem2-2-eq2}, \eqref{theorem1-eq1}, \eqref{theorem1-2-eq1}, \eqref{theorem1-2-eq2}, and \eqref{theorem1-2-eq9}--\eqref{theorem1-2-eq13} gives
\begin{flalign}
\nonumber
&\;\;\;\;\;{\mathbf{E}}\Big[\sum\limits_{i = 1}^n {\sum\limits_{t = 1}^T {\| {{{[ {{g_{i,t}}( {{x_{i,t}}} )} ]}_ + }} \|} }\Big] \\
\nonumber
& \le \frac{{2n{R{{( \mathbb{X} )}^2}}}}{{{\varsigma _s}{\gamma _0}}} + \frac{{2nF{T^{1 - c}}}}{{{\varsigma _s}{\gamma _0}( {1 - c} )}} + \frac{{4n{{ {p^2{G_1^2} } }}\log ( T )}}{{ {\varsigma _s}{\gamma _0}}} + \frac{{n R{{( \mathbb{X} )}^2}\log ( T )}}{{2{\varsigma _s}{\gamma _0}}} + \frac{{n r{{( \mathbb{X} )}^2}\log ( T )}}{{2{\varsigma _s}{\gamma _0}}}\\
&\;\; + \frac{{2n {{G_1} } R( \mathbb{X} )\log ( T )}}{{{\varsigma _s}{\gamma _0}}} + \frac{{2n{G_1}r( \mathbb{X} )\log ( T )}}{{{\varsigma _s}{\gamma _0}}}. \label{theorem2-2-eq5}
\end{flalign}

From \eqref{theorem1-eq1}, \eqref{lemma6-eqc}, \eqref{theorem2-2-eq5}, we have
\begin{flalign}
\nonumber
&\;\;\;\;\;{\mathbf{E}}[{\rm{Net}\mbox{-}\rm{CCV}}( T )] \\
\nonumber
& \le n{G_2}{\tilde{\varepsilon} _1} + \frac{{{\varepsilon _3}{T^{1 - c}}}}{{1 - c}} + \frac{{2n{R{{( \mathbb{X} )}^2}}{\varepsilon _4}}}{{{\varsigma _s}{\gamma _0}}} + \frac{{2nF{\varepsilon _4}{T^{1 - c}}}}{{{\varsigma _s}{\gamma _0}( {1 - c} )}} + \frac{{4n{{ {p^2{G_1^2} } }}{\varepsilon _4}\log ( T )}}{{ {\varsigma _s}{\gamma _0}}} + \frac{{n R{{( \mathbb{X} )}^2}{\varepsilon _4}\log ( T )}}{{2{\varsigma _s}{\gamma _0}}} \\
&\;\; + \frac{{n r{{( \mathbb{X} )}^2}{\varepsilon _4}\log ( T )}}{{2{\varsigma _s}{\gamma _0}}} + \frac{{2n {{G_1} } R( \mathbb{X} ){\varepsilon _4}\log ( T )}}{{{\varsigma _s}{\gamma _0}}} + \frac{{2n{G_1}r( \mathbb{X} ){\varepsilon _4}\log ( T )}}{{{\varsigma _s}{\gamma _0}}}. \label{theorem2-2-eq6}
\end{flalign}

When  $c = ( {1/2,1} )$, combining \eqref{theorem2-2-eq2}, \eqref{theorem1-eq1}, \eqref{theorem1-2-eq1}, \eqref{theorem1-2-eq2}, and \eqref{theorem1-2-eq15}--\eqref{theorem1-2-eq19} gives
\begin{flalign}
\nonumber
&\;\;\;\;\;{\mathbf{E}}\Big[\sum\limits_{i = 1}^n {\sum\limits_{t = 1}^T {\| {{{[ {{g_{i,t}}( {{x_{i,t}}} )} ]}_ + }} \|} }\Big] \\
\nonumber
& \le \frac{{2n{R{{( \mathbb{X} )}^2}}}}{{{\varsigma _s}{\gamma _0}}} + \frac{{2nF{T^{1 - c}}}}{{{\varsigma _s}{\gamma _0}( {1 - c} )}} + \frac{{2n{{ {p^2{G_1^2}} }}Q}}{{ {\varsigma _s}{\gamma _0}}} + \frac{{n R{{( \mathbb{X} )}^2}Q}}{{4{\varsigma _s}{\gamma _0}}} + \frac{{n r{{( \mathbb{X} )}^2}Q}}{{4{\varsigma _s}{\gamma _0}}} + \frac{{n {{G_1} } R( \mathbb{X} )Q}}{{{\varsigma _s}{\gamma _0}}}\\
&\;\;  + \frac{{n{G_1}r( \mathbb{X} )Q}}{{{\varsigma _s}{\gamma _0}}}. \label{theorem2-2-eq7}
\end{flalign}

From \eqref{theorem1-eq1}, \eqref{lemma6-eqc}, \eqref{theorem2-2-eq7}, we have
\begin{flalign}
\nonumber
&\;\;\;\;\;{\mathbf{E}}[{\rm{Net}\mbox{-}\rm{CCV}}( T )] \\
\nonumber
& \le n{G_2}{\tilde{\varepsilon} _1} + \frac{{{\varepsilon _3}{T^{1 - c}}}}{{1 - c}} + \frac{{2n{R{{( \mathbb{X} )}^2}}{\varepsilon _4}}}{{{\varsigma _s}{\gamma _0}}} + \frac{{2nF{\varepsilon _4}{T^{1 - c}}}}{{{\varsigma _s}{\gamma _0}( {1 - c} )}} + \frac{{2n{{ {p^2{G_1^2} } }}Q{\varepsilon _4}}}{{ {\varsigma _s}{\gamma _0}}} + \frac{{n R{{( \mathbb{X} )}^2}Q{\varepsilon _4}}}{{4{\varsigma _s}{\gamma _0}}} \\
&\;\; + \frac{{n r{{( \mathbb{X} )}^2}Q{\varepsilon _4}}}{{4{\varsigma _s}{\gamma _0}}} + \frac{{n {{G_1} } R( \mathbb{X} )Q{\varepsilon _4}}}{{{\varsigma _s}{\gamma _0}}} + \frac{{n{G_1}r( \mathbb{X} )Q{\varepsilon _4}}}{{{\varsigma _s}{\gamma _0}}}. \label{theorem2-2-eq8}
\end{flalign}

It follows from \eqref{theorem2-2-eq4}, \eqref{theorem2-2-eq6}, and \eqref{theorem2-2-eq8} that \eqref{theorem2-eq2} holds.

\hspace{-3mm}\emph{D. Proof of Theorem~3}

We will show that \eqref{Theorem3-eq1}--\eqref{Theorem3-eq3} in Theorem~3 hold in the following, respectively.

Since Assumption~6 holds, \eqref{lemma4-eq8} can be replaced by
\begin{flalign}
{\hat f_{i,t}}( {{x_{i,t}}} ) - {\hat f_{i,t}}( {\hat y} ) \le {\mathbf{E}_{{\mathfrak{U}_t}}}[ {p{G_1}\| {\varepsilon _{i,t}^z} \| + \langle {\hat \partial {f_{i,t}}( {{x_{i,t}}} ),{z_{i,t + 1}} - \hat y} \rangle  - \frac{\mu }{2}{\| {\hat y - {x_{i,t}}} \|^2}} ]. \label{theorem3-eq1}
\end{flalign}

Note that different from \eqref{lemma4-eq8}, there exists an extra term $\frac{\mu }{2}{\| {\hat y - {x_{i,t}}} \|^2}$ in \eqref{theorem3-eq1}, and thus \eqref{lemma6-eqa} can be replaced by
\begin{flalign}
\nonumber
&\;\;\;\;\;{\mathbf{E}}\Big[\frac{1}{n}\sum\limits_{i = 1}^n {\sum\limits_{t = 1}^T {{f_t}( {{x_{i,t}}} )} }  - \sum\limits_{t = 1}^T {{f_t}( y )}\Big] \\
& \le {\varepsilon _1}{G_1} + {n_T} + {\mathbf{E}}[\frac{1}{n}\sum\limits_{i = 1}^n {\sum\limits_{t = 1}^T {\big( {{\Delta _{i,t}}(\hat y) - \frac{\mu }{2}{\| {\hat y - {x_{i,t}}} \|^2}} \big)} }]. \label{theorem3-eq2}
\end{flalign}

Moreover, \eqref{lemma5-eqb} can be replaced by
\begin{flalign}
\sum\limits_{t = 1}^T {\sum\limits_{i = 1}^n {\frac{1}{2}} } \Big( {\frac{{q_{i,t + 1}^T{g_{i,t}}( {{x_{i,t}}} )}}{{{\gamma _t}}} + \frac{{ {\mathbf{E}_{{\mathfrak{U}_t}}}[{{\| {\varepsilon _{i,t}^z} \|}^2}]}}{{2{\gamma _0}}}} \Big) \le {h_T}( y ) + {{\tilde h}_T}( {\hat y} ) + {{\hat h}_T}( {\hat y} ), \label{theorem3-eq3}
\end{flalign}
where
\begin{flalign}
\nonumber
{{\hat h}_T}( {\hat y} ) =  - \sum\limits_{i = 1}^n {\sum\limits_{t = 1}^T {\frac{{\mu {\mathbf{E}_{{\mathfrak{U}_t}}}[{{{\| {\hat y - {x_{i,t}}} \|}^2}}]}}{{{2\gamma _t}}}} }.
\end{flalign}

As a result, \eqref{lemma6-eqb} can be replaced by
\begin{flalign}
{\mathbf{E}}\Big[\frac{1}{n}\sum\limits_{i = 1}^n {\sum\limits_{t = 1}^T {\| {{{[ {{g_t}( {{x_{i,t}}} )} ]}_ + }} \|} }\Big]  \le \sqrt {2G_2^2{{\tilde \varepsilon }_3}T + {\varepsilon _2}T{\mathbf{E}}\big[\big( {{{\tilde h}_T}( {\hat y} ) + {{\hat h}_T}( {\hat y} )} \big)\big]}. \label{theorem3-eq4}
\end{flalign}

\noindent (i)
We have
\begin{flalign}
\nonumber
&\;\;\;\;\;\frac{1}{n}\sum\limits_{t = 1}^T {\sum\limits_{i = 1}^n {\big( {{\Delta _{i,t}}( {\hat y} ) { - \frac{\mu }{2}{{\| {\hat y - {x_{i,t}}} \|}^2}}} \big)} } \\
\nonumber
& = \frac{1}{n}\sum\limits_{t = 1}^T {\sum\limits_{i = 1}^n } \Big( {{\frac{1}{{2{\alpha _t}}}}{\mathbf{E}_{{\mathfrak{U}_t}}}[{{\| {\hat y - {x_{i,t}}} \|}^2} - {{\| {\hat y - {x_{i,t + 1}}} \|}^2}] - \frac{\mu }{2}{{\| {\hat y - {x_{i,t}}} \|}^2}} \Big)\\
\nonumber
& = \frac{1}{{2n}}\sum\limits_{t = 1}^T {\sum\limits_{i = 1}^n {\Big( {\frac{1}{{{\alpha _{t - 1}}}}{\mathbf{E}_{{\mathfrak{U}_t}}}[{{\| {\hat y - {x_{i,t}}} \|}^2}] - \frac{1}{{{\alpha _t}}}{\mathbf{E}_{{\mathfrak{U}_t}}}[{{\| {\hat y - {x_{i,t + 1}}} \|}^2}] + \Big( {\frac{1}{{{\alpha _t}}} - \frac{1}{{{\alpha _{t - 1}}}} - \mu} \Big){\mathbf{E}_{{\mathfrak{U}_t}}}[{{\| {\hat y - {x_{i,t}}} \|}^2}]} \Big)} }  \\
& \le \frac{1}{{2n}}\sum\limits_{i = 1}^n {\Big( {\frac{1}{{{\alpha _0}}}{\mathbf{E}_{{\mathfrak{U}_t}}}[{{\| {\hat y - {x_{i,1}}} \|}^2}] - \frac{1}{{{\alpha _T}}}{\mathbf{E}_{{\mathfrak{U}_t}}}[{{\| {\hat y - {x_{i,T + 1}}} \|}^2}] + \sum\limits_{t = 1}^T {\Big( {\frac{1}{{{\alpha _t}}} - \frac{1}{{{\alpha _{t - 1}}}} - \mu} \Big){\mathbf{E}_{{\mathfrak{U}_t}}}[{{\| {\hat y - {x_{i,t}}} \|}^2}}] } \Big)}. \label{theorem3-eq5}
\end{flalign}

Denote
\begin{flalign}
\nonumber
{\varepsilon _5} = \Big\lceil {{{\Big( {\frac{1}{\mu }} \Big)}^{\frac{1}{{1 - c}}}}} \Big\rceil.
\end{flalign}

From \eqref{theorem1-eq1}, we have
\begin{flalign}
\frac{1}{{{\alpha _{t + 1}}}} - \frac{1}{{{\alpha _t}}} - \mu  = \frac{{t + 1}}{{{{( {t + 1} )}^{1 - c}}}} - \frac{t}{{{t^{1 - c}}}} - \mu  < \frac{1}{{{t^{1 - c}}}} - \mu \le 0, \forall t \ge {\varepsilon _5}. \label{theorem3-eq6}
\end{flalign}

Choosing $y = {x^*} \in {\mathcal{X}_T}$, and combining \eqref{ass1-eq1}, \eqref{theorem1-1-eq2}--\eqref{theorem1-1-eq6}, \eqref{theorem3-eq2}, and \eqref{theorem3-eq5}--\eqref{theorem3-eq6} yields
\begin{flalign}
\nonumber
{\mathbf{E}}[ {{\rm{Net} \mbox{-} \rm{Reg}}( T )} ] &\le {\varepsilon _1} {{G_1} } + \frac{{{{ {{G_1^2} } }}\tilde \varepsilon _2^2{T^{1 - c}}}}{{ ( {1 - c} )}} + \frac{{2{{ {p^2{G_1^2} } }}{T^{1 - c}}}}{{ ( {1 - c} )}} + \frac{{ R{{( \mathbb{X} )}^2}{T^{1 - c}}}}{{4( {1 - c} )}} + \frac{{ r{(\mathbb{X})^2}{T^{1 - c}}}}{{4( {1 - c} )}}\\
\nonumber
&\;\;  + \frac{{ {{G_1}} R( \mathbb{X} ){T^{1 - c}}}}{{1 - c}} + \frac{{{G_1}r( \mathbb{X} ){T^{1 - c}}}}{{1 - c}} + \frac{1}{{2n{\alpha _0}}}\sum\limits_{i = 1}^n {\mathbf{E}_{{\mathfrak{U}_t}}}[{{{\| {\hat y - {x_{i,1}}} \|}^2}}] \\
\nonumber
&\;\; + \frac{1}{n}\sum\limits_{i = 1}^n {\sum\limits_{t = 1}^{{\varepsilon _5}} {\Big( {\frac{1}{{{\alpha _t}}} - \frac{1}{{{\alpha _{t - 1}}}} - \mu } \Big)} } {\mathbf{E}_{{\mathfrak{U}_t}}}[{\| {\hat y - {x_{i,t}}} \|^2}] \\
\nonumber
&\le {\varepsilon _1} {{G_1} }  + \frac{{{{ {{G_1^2} } }}\tilde \varepsilon _2^2{T^{1 - c}}}}{{ ( {1 - c} )}} + \frac{{2{{ {p^2{G_1^2} } }}{T^{1 - c}}}}{{ ( {1 - c} )}} + \frac{{ R{{( \mathbb{X} )}^2}{T^{1 - c}}}}{{4( {1 - c} )}} + \frac{{ r{(\mathbb{X})^2}{T^{1 - c}}}}{{4( {1 - c} )}}\\
&\;\;  + \frac{{ {{G_1} } R( \mathbb{X} ){T^{1 - c}}}}{{1 - c}} + \frac{{{G_1}r( \mathbb{X} ){T^{1 - c}}}}{{1 - c}}  + 4{\varepsilon _5}{[ {1 - \mu } ]_ + }{R{{( \mathbb{X} )}^2}}.
\end{flalign}
Hence, \eqref{Theorem3-eq1} holds.

\noindent (ii)
From \eqref{theorem1-2-eq8}, \eqref{theorem1-2-eq14}, \eqref{theorem1-2-eq20}, and \eqref{theorem3-eq3}--\eqref{theorem3-eq4}, we have \eqref{Theorem3-eq2}.

\noindent (iii)
From \eqref{theorem2-2-eq4}, \eqref{theorem2-2-eq6}, and \eqref{theorem2-2-eq8}, and \eqref{theorem3-eq3}--\eqref{theorem3-eq4}, we have \eqref{Theorem3-eq3}.

\hspace{-3mm}\emph{E. Proof of Theorem~4}

We will show that \eqref{theorem4-eq2}--\eqref{theorem4-eq4} in Theorem~4 hold in the following, respectively.

Note that \eqref{theorem3-eq2}--\eqref{theorem3-eq4} still hold.

\noindent (i)
From \eqref{theorem4-eq1} and \eqref{theorem3-eq5}, we have
\begin{flalign}
\frac{1}{n}\sum\limits_{t = 1}^T {\sum\limits_{i = 1}^n {\big( {{\Delta _{i,t}}( {\hat y} ) { - \frac{\mu }{2}{{\| {\hat y - {x_{i,t}}} \|}^2}}} \big)} } 
\le 0. \label{theorem4-1-eq1}
\end{flalign}

From \eqref{theorem4-eq1}, we have
\begin{flalign}
\sum\limits_{t = 1}^T {{\alpha _t}}  &= \sum\limits_{t = 2}^T {\frac{1}{{\mu t}}}  + \frac{1}{\mu } \le \int_1^T {\frac{1}{{\mu t}}} dt + \frac{1}{\mu } \le \frac{1}{\mu }\big( {\log ( T ) + 1} \big), \label{theorem4-1-eq2}\\
\sum\limits_{t = 1}^T {\frac{{\xi _t^2}}{{{\alpha _t}}}} & = \sum\limits_{t = 1}^T {{\alpha _t}}  \le \frac{1}{\mu }\big( {\log ( T ) + 1} \big), \label{theorem4-1-eq3}\\
\sum\limits_{t = 1}^T {\frac{{\delta _t^2}}{{{\alpha _t}}}} & = r{(\mathbb{X})^2}\sum\limits_{t = 1}^T {{\alpha _t}}  \le \frac{{r{{(\mathbb{X})}^2}}}{\mu }\big( {\log ( T ) + 1} \big), \label{theorem4-1-eq4}\\
\sum\limits_{t = 1}^T {{\xi _t}}  &= \sum\limits_{t = 1}^T {{\alpha _t}}  \le \frac{1}{\mu }\big( {\log ( T ) + 1} \big), \label{theorem4-1-eq5}\\
\sum\limits_{t = 1}^T {{\delta _t}}  &= r(\mathbb{X})\sum\limits_{t = 1}^T {{\alpha _t}}  \le \frac{{r(\mathbb{X})}}{\mu }\big( {\log ( T ) + 1} \big). \label{theorem4-1-eq6}
\end{flalign}

Choosing $y = {x^*} \in {\mathcal{X}_T}$, and combining \eqref{theorem3-eq2}, and \eqref{theorem4-1-eq1}--\eqref{theorem4-1-eq6} yields
\begin{flalign}
\nonumber
{\mathbf{E}}[ {{\rm{Net} \mbox{-} \rm{Reg}}( T )} ] &\le {\varepsilon _1} {{G_1} }  + \frac{{{{{G_1^2} }}\tilde \varepsilon _2^2\big( {\log ( T ) + 1} \big)}}{{ \mu }} + \frac{{2{{p^2{G_1^2} }}\big( {\log ( T ) + 1} \big)}}{{ \mu }}  \\
\nonumber
&\;\; + \frac{{ R{{(\mathbb{X})}^2}\big( {\log ( T ) + 1} \big)}}{{4\mu }} + \frac{{ r{{(\mathbb{X})}^2}\big( {\log ( T ) + 1} \big)}}{{4\mu }}  \\
&\;\; + \frac{{{G_1} R(\mathbb{X})\big( {\log ( T ) + 1} \big)}}{\mu }  + \frac{{{G_1}r(\mathbb{X})\big( {\log ( T ) + 1} \big)}}{\mu }. \label{theorem4-1-eq7}
\end{flalign}
Hence, \eqref{theorem4-eq2} holds.

\noindent (ii)
From \eqref{theorem4-eq1}, we have
\begin{flalign}
&\sum\limits_{t = 1}^T {\frac{1}{{{\gamma _t}}} = \sum\limits_{t = 1}^T {\frac{{{\alpha _t}}}{{{\gamma _0}}}}  \le } \frac{1}{{{\gamma _0}\mu }}\big( {\log ( T ) + 1} \big), \label{theorem4-2-eq1}\\
&\sum\limits_{t = 1}^T {\alpha _t^2}  = \frac{1}{{{\mu ^2}}}\sum\limits_{t = 1}^T {\frac{1}{{{t^2}}}}  = \frac{1}{{{\mu ^2}}}\Big( {\sum\limits_{t = 2}^T {\frac{1}{{{t^2}}}}  + 1} \Big) \le \frac{1}{{{\mu ^2}}}\Big( {\int_{t = 2}^T {\frac{1}{{{t^2}}}dt}  + 1} \Big) \le \frac{2}{{{\mu ^2}}}, \label{theorem4-2-eq2}\\
&\sum\limits_{t = 1}^T {\frac{1}{{\gamma _t^2}}}  = \sum\limits_{t = 1}^T {\frac{{\alpha _t^2}}{{\gamma _0^2}}}  \le \frac{2}{{\gamma _0^2{\mu ^2}}}, \label{theorem4-2-eq3}\\
&\sum\limits_{t = 1}^T {\xi _t^2}  = \sum\limits_{t = 1}^T {\alpha _t^2}  \le \frac{2}{{{\mu ^2}}}, \label{theorem4-2-eq4}\\
&\sum\limits_{t = 1}^T {\delta _t^2}  = r{(\mathbb{X})^2}\sum\limits_{t = 1}^T {\alpha _t^2}  \le \frac{{2r{{(\mathbb{X})}^2}}}{{{\mu ^2}}}, \label{theorem4-2-eq5}\\
&\sum\limits_{t = 1}^T {\frac{{{\xi _t}}}{{{\gamma _t}}}}  = \frac{1}{{{\gamma _0}}}\sum\limits_{t = 1}^T {\alpha _t^2}  \le \frac{2}{{{\gamma _0}{\mu ^2}}}, \label{theorem4-2-eq6}\\
&\sum\limits_{t = 1}^T {\frac{{{\delta _t}}}{{{\gamma _t}}}}  = \frac{{r(\mathbb{X})}}{{{\gamma _0}}}\sum\limits_{t = 1}^T {\alpha _t^2}  \le \frac{{2r(\mathbb{X})}}{{{\gamma _0}{\mu ^2}}}, \label{theorem4-2-eq7}\\
&\sum\limits_{i = 1}^n \frac{{\mathbf{E}_{{\mathfrak{U}_t}}}[{{{\| {\hat y - {x_{i,1}}} \|}^2}}]}{{{2\gamma _0}}} + {{\hat h}_T}( {\hat y} ) \le \sum\limits_{i = 1}^n \frac{{\mathbf{E}_{{\mathfrak{U}_t}}}[{{{\| {\hat y - {x_{i,1}}} \|}^2}}]}{{{2\gamma _0}}}  - \sum\limits_{i = 1}^n \frac{\mu{\mathbf{E}_{{\mathfrak{U}_t}}}[{{{\| {\hat y - {x_{i,1}}} \|}^2}}]}{{{2\gamma _1}}}  = 0. \label{theorem4-2-eq8}
\end{flalign}

From \eqref{theorem3-eq4}, and \eqref{theorem4-2-eq1}--\eqref{theorem4-2-eq8}, we have
\begin{flalign}
\nonumber
&\;\;\;\;\;{\mathbf{E}}\Big[{\Big( {\frac{1}{n}\sum\limits_{i = 1}^n {\sum\limits_{t = 1}^T {\| {{{[ {{g_t}( {{x_{i,t}}} )} ]}_ + }} \|} } } \Big)^2}\Big] \\
\nonumber
&\le 2G_2^2{{\tilde \varepsilon }_3}T + \frac{{2nF{\varepsilon _2}T\big( {\log ( T ) + 1} \big)}}{{{\gamma _0}\mu }} + \frac{{4n{{p^2{G_1^2} }}{\varepsilon _2}T}}{{ {\gamma _0}{\mu ^2}}} + \frac{{n R{{(\mathbb{X})}^2}{\varepsilon _2}T}}{{2{\gamma _0}{\mu ^2}}} + \frac{{n r{{(\mathbb{X})}^2}{\varepsilon _2}T}}{{2{\gamma _0}{\mu ^2}}}\\
&\;\; + \frac{{2n{G_1} R(\mathbb{X}){\varepsilon _2}T}}{{{\gamma _0}{\mu ^2}}} + \frac{{2n{G_1}r(\mathbb{X}){\varepsilon _2}T}}{{{\gamma _0}{\mu ^2}}}. \label{theorem4-2-eq9}
\end{flalign}

It follows from \eqref{theorem4-2-eq9} that \eqref{theorem4-eq3} holds.

\noindent (iii)
Choosing $y = {x_s}$ in \eqref{theorem3-eq3}, and using \eqref{Algorithm1-eq2} and \eqref{theorem2-2-eq1} gives
\begin{flalign}
{\varsigma _s}\sum\limits_{t = 1}^T {\sum\limits_{i = 1}^n {\| {{{[ {{g_{i,t}}( {{x_{i,t}}} )} ]}_ + }} \|} }  \le {{\tilde h}_T}( {\hat y} ) + {{\hat h}_T}( {\hat y} ). \label{theorem4-3-eq1}
\end{flalign}
Combining \eqref{theorem4-2-eq1}--\eqref{theorem4-2-eq8}, and \eqref{theorem4-3-eq1} gives
\begin{flalign}
\nonumber
&\;\;\;\;\;{\mathbf{E}}\Big[\sum\limits_{i = 1}^n {\sum\limits_{t = 1}^T {\| {{{[ {{g_{i,t}}( {{x_{i,t}}} )} ]}_ + }} \|} }\Big] \\
& \le \frac{{2nF\big( {\log ( T ) + 1} \big)}}{{{\gamma _0}\mu {\varsigma _s}}} + \frac{{4n{{p^2{G_1^2}}}}}{{ {\gamma _0}{\mu ^2}{\varsigma _s}}} + \frac{{n R{{(\mathbb{X})}^2}}}{{2{\gamma _0}{\mu ^2}{\varsigma _s}}} + \frac{{n r{{(\mathbb{X})}^2}}}{{2{\gamma _0}{\mu ^2}{\varsigma _s}}} + \frac{{2n{G_1}R(\mathbb{X})}}{{{\gamma _0}{\mu ^2}{\varsigma _s}}} + \frac{{2n{G_1}r(\mathbb{X})}}{{{\gamma _0}{\mu ^2}{\varsigma _s}}}. \label{theorem4-3-eq2}
\end{flalign}

From \eqref{theorem4-eq1}, \eqref{lemma6-eqc}, \eqref{theorem4-3-eq1}, and \eqref{theorem4-3-eq2}, we have
\begin{flalign}
\nonumber
&\;\;\;\;\;{\mathbf{E}}[{\rm{Net}\mbox{-}\rm{CCV}}( T )] \\
\nonumber
& \le n{G_2}{\tilde{\varepsilon} _1} + \frac{{{\varepsilon _3}\big( {\log ( T ) + 1} \big)}}{\mu } + \frac{{2nF{\varepsilon _4}\big( {\log ( T ) + 1} \big)}}{{{\gamma _0}\mu {\varsigma _s}}} + \frac{{4n{{p^2{G_1^2} }}{\varepsilon _4}}}{{ {\gamma _0}{\mu ^2}{\varsigma _s}}} + \frac{{n R{{(\mathbb{X})}^2}{\varepsilon _4}}}{{2{\gamma _0}{\mu ^2}{\varsigma _s}}} \\
&\;\; + \frac{{n r{{(\mathbb{X})}^2}{\varepsilon _4}}}{{2{\gamma _0}{\mu ^2}{\varsigma _s}}} + \frac{{2n{G_1}R(\mathbb{X}){\varepsilon _4}}}{{{\gamma _0}{\mu ^2}{\varsigma _s}}} + \frac{{2n{G_1}r(\mathbb{X}){\varepsilon _4}}}{{{\gamma _0}{\mu ^2}{\varsigma _s}}}. \label{theorem4-4-eq3}
\end{flalign}

It follows from \eqref{theorem4-4-eq3} that \eqref{theorem4-eq4} holds.
\bibliographystyle{IEEEtran}
\bibliography{reference_online}

\end{document}